\documentclass{amsart}
\long\def\sidebyside#1#2{%
 \hbox to\textwidth{\vtop{\hsize=.6\textwidth%

 \advance\hsize by -.32\columnsep
\parindent=0pt
\centering

 #1\vskip1sp}\hskip\columnsep\vtop{\hsize=.6\textwidth%
 \advance\hsize by -.32\columnsep
\parindent=0pt
\centering
#2

}\hfill}}

\usepackage{amsthm, amssymb, amsfonts}
\usepackage{psfrag}
\usepackage{graphicx}
\normalsize
\theoremstyle{definition}
\newtheorem{Thm}{Theorem}

\newtheorem{ddefinition}[Thm]{Definition}
\newtheorem{ttheorem}[Thm]{Theorem}

\begin{document}
\newcommand {\ep}{\epsilon}
\newcommand {\gmupvb}{\gamma_{\om\ub\op\vb}^{\phantom{1}}}
\newcommand {\op}{\mathbf{\oplus}} 
\newcommand {\om}{\mathbf{\ominus}} 
\newcommand {\od}{\mathbf{\otimes}}   
\newcommand {\sqp}{\boxplus}
\newcommand {\sqm}{\boxminus}
\newcommand {\gyrab}{\gyr[a,b]}
\newcommand {\gyr}{{\rm gyr}}
\newcommand {\gyrabb}{\gyr[\ab,\bb]}
\newcommand {\gyruvb}{\gyr[\ub,\vb]}
\newcommand {\gyrvub}{\gyr[\vb,\ub]}
\newcommand {\Gyr}{{\rm Gyr}}
\newcommand {\ccdot}{\mathbf{\cdot }}
\newcommand {\ab}{\mathbf{a}}
\newcommand {\bb}{\mathbf{b}}
\newcommand {\cb}{\mathbf{c}}
\newcommand {\ub}{\mathbf{u}}
\newcommand {\vb}{\mathbf{v}}
\newcommand {\wb}{\mathbf{w}}
\newcommand {\xb}{\mathbf{x}}
\newcommand {\vs}{\mathbb{V}_s}
\newcommand {\gub}{\gamma_{\ub}^{\phantom{1}}}
\newcommand {\gvb}{\gamma_{\vb}^{\phantom{1}}}
\newcommand {\gwb}{\gamma_{\wb}^{\phantom{1}}}
\newcommand {\gubs}{\gamma_{\ub}^2}
\newcommand {\gvbs}{\gamma_{\vb}^2}
\newcommand {\gwbs}{\gamma_{\wb}^2}
\newcommand {\gupvb}{\gamma_{\ub\op\vb}^{\phantom{1}}}
\newcommand {\gupvbs}{\gamma_{\ub\op\vb}^2}
\newcommand {\gmupvbs}{\gamma_{\om\ub\op\vb}^2}
\newcommand {\Rb}{\mathbb{R}}
\newcommand {\Rn}{\Rb^n}
\newcommand {\Rt}{\Rb^3}
\newcommand {\Rcn}{{\Rb}_{c}^{n}}
\newcommand {\Rct}{{\Rb}_{c}^{3}}
\newcommand {\Rsn}{{\Rb}_{s}^{n}}
\newcommand {\Rst}{{\Rb}_{s}^{3}}
\newcommand {\Rtwo}{\Rb^2}
\newcommand {\Rctwo}{{\Rb}_{s}^2}
\newcommand {\Rstwo}{{\Rb}_{s}^2}
\newcommand {\uvc}{\displaystyle\frac{\lower.6ex \hbox {$\ub\ccdot\vb$}}{c^2}}
\newcommand {\uvs}{\displaystyle\frac{\lower.6ex \hbox {$\ub\ccdot\vb$}}{s^2}}
\newcommand {\unpuvc}{ \lower.6ex \hbox {$1 + \uvc$} }
\newcommand {\zerb}{\mathbf{0}}
\newcommand {\bz}{\mathbf{0}}
\newcommand {\Aut}{{\rm Aut}}
\newcommand {\ro}{r_{_1}}
\newcommand {\rt}{r_{_2}}
\newcommand {\trace}{{\rm trace}}
\newcommand {\ggub}{\sqrt{\gamma_{\ub}^2-1}}
\newcommand {\ggvb}{\sqrt{\gamma_{\vb}^2-1}}
\newcommand {\ggupvb}{\sqrt{\gamma_{\ub\op\vb}^2-1}}
\newcommand {\omegab}{\boldsymbol{\omega}}
\newcommand {\cmt}{c^{-2}}
\newcommand {\tonxb}{\begin{pmatrix} t \\ \xb \end{pmatrix}}
\newcommand {\tonxbprime}{\begin{pmatrix}t^\prime\\ \xb^\prime\end{pmatrix}}
\newcommand {\half}{\textstyle\frac{1}{2}}

\vskip.8cm
\addtocounter{page}{-8}
\pagenumbering{arabic}
\begin{center}
\huge{
Gyrations: The Missing Link Between
Classical Mechanics with its Underlying Euclidean Geometry
and
Relativistic Mechanics with its Underlying Hyperbolic Geometry
     }
\end{center}
\begin{center}
Abraham A. Ungar\\
Department of Mathematics\\
North Dakota State University\\
Fargo, ND 58105, USA\\
Email: abraham.ungar@ndsu.edu\\
\vspace{0.4cm}
Dedicated to the 80th Anniversary of Steve Smale\\[8pt]
\end{center}

\begin{quotation}
{\bf ABSTRACT}\phantom{OO}
Being neither commutative nor associative, Einstein velocity addition
of relativistically admissible velocities
gives rise to gyrations. Gyrations, in turn, measure the
extent to which Einstein addition deviates from commutativity and from associativity.
Gyrations are geometric automorphisms abstracted from the relativistic mechanical effect
known as Thomas precession.
Thomas precession is a peculiar special relativistic space rotation that stems from
Einstein velocity addition so that, in turn, Einstein velocity addition is
algebraically and geometrically regulated by Thomas precession.
Gyrations give rise to {\it gyroalgebra} and {\it gyrogeometry}, elegant nonassociative
algebra and geometry in which groups and vector spaces become gyrogroups and
gyrovector spaces and Euclidean geometry becomes hyperbolic geometry. Indeed,
gyrovector spaces form the algebraic setting for
the hyperbolic geometry of Bolyai and Lobachevsky just as
vector spaces form the algebraic setting for Euclidean geometry.
As such, gyrations provide the missing link between Einstein's special theory of relativity
and the hyperbolic geometry of Bolyai and Lobachevsky.

In an exhaustive review of the vast literature on Thomas precession,
G.B.~Malykin points out that many studies of Thomas precession
are erroneous for various reasons.
In particular, Malykin's review reveals that there is no agreement among
Thomas precession explorers on whether or not the Thomas precession angle
of rotation and its generating angle have opposite signs.
The mission of this article is to employ the study of gyrations and their geometry to
the understanding of Thomas precession and, hence,
to put Thomas precession misconceptions to rest.
This article is dedicated to the 80th Anniversary of Steve Smale,
whose work in geometric mechanics is well known.
\end{quotation}

\section{Introduction}
\label{secc1}

It has been the lifetime desire of Steve Smale to improve our understanding
of geometric mechanics \cite{smale2000}, a theory that fits extraordinarily well into
dynamical systems framework, as he explains in his 1967 survey article \cite{smale1967}.
An overview of his involvement with geometric mechanics, without entering
into technical details, is presented by Marsden in \cite{marsden93}.
The impact of Einstein's work \cite{rassiaseins82} led smale to remark
in \cite[p.~365]{smale2000} that relativity theory respects classical mechanics since
``Einstein worked from a very deep understanding of Newtonian theory''.
The present article on the hyperbolic geometric interpretation of
the relativistic mechanical effect known as Thomas precession
is therefore dedicated to the 80th Anniversary of Steve Smale for
his leadership and commitment to excellence in the field of geometric mechanics.

Thomas precession of Einstein's special theory of relativity is a physical
realization of the abstract gyration. The latter, in turn, is
an automorphism (defined in Def.~\ref{defroupx})
that provides the missing link between
Einstein's special theory of relativity
and the hyperbolic geometry of Bolyai and Lobachevsky.
Named after Llewellyn Hilleth Thomas (1902--1992) who discovered its physical significance
in 1926 \cite{belloni86,thomas26,thomas27}, Thomas precession
is a special relativistic kinematic effect that regulates Einstein velocity addition
both algebraically and geometrically \cite{mybook01}.
In an exhaustive review of the vast literature on Thomas precession
\cite{malykin06}, G.B.~Malykin emphasizes the importance of the frequency of the precession,
pinpointing related erroneous results that are common in the literature.

Accordingly, a study of Thomas precession in terms of its underlying
hyperbolic geometry and elegant algebra is presented here in order to clarify
the concept of Thomas precession.
Thomas precession is an important special relativistic rotation
that results from the nonassociativity of Einstein velocity addition and, hence,
does not exist classically.
Indeed, it was discovered in 1988 \cite{parametrization} that
Thomas precession regulates Einstein velocity addition, endowing it with
a rich algebraic structure. As such,
Thomas precession admits extension by abstraction, in which
{\it precession} becomes {\it gyration}. The latter, in turn, gives rise to
two new algebraic structures called a gyrogroup and a gyrovector space,
thus introducing new realms to explore.
The basic importance of gyrations is emphasized in gyrolanguage, where we
prefix a gyro to any term that describes a
concept in Euclidean geometry and in associative algebra
to mean the analogous concept in hyperbolic geometry and nonassociative algebra.

Gyrovector spaces turn out to form the algebraic setting for the
hyperbolic geometry of Bolyai and Lobachevsky
just as vector spaces form the algebraic setting for Euclidean geometry
\cite{service,mybook01,walterrev2002}.
This discovery resulted in the extension by abstraction
of Thomas precession into gyration, and in subsequent studies presented in
several books
\cite{mybook01,mybook02,mybook03,mybook04,mybook06,mybook05}
reviewed, for instance, in
\cite{walterrev2002} and \cite{rassiasrev2010}.
The hyperbolic geometric character of Thomas precession is emphasized here
by the observation that it can be interpreted as the defect of a related
hyperbolic triangle.

Following Malykin's observations in \cite{malykin06}, there is a need to demonstrate that
our study of Thomas precession does lead to the correct Thomas precession angular velocity.
Accordingly, in this article we review the studies in
\cite{mybook01,mybook02,mybook03,mybook04,mybook06,mybook05}
of both Thomas precession and its abstract version, gyration.
Based on the review we derive the correct Thomas precession angular velocity,
illustrated in Fig.~\ref{fig115}.

A boost is a Lorentz transformation without rotation \cite{quasi91}.
The Thomas precession angle $\ep$ is generated by the application of
two successive boosts with velocity parameters, say, $\ub$ and $\vb$.
The angle $\theta$ between $\ub$ and $\vb$ is the generating angle of
the resulting Thomas precession angle $\ep$, shown in Fig.~\ref{fig115}.

An important question about the Thomas precession angle $\ep$ and its
generating angle $\theta$ is whether or not $\ep$ and $\theta$ have
equal signs. According to Malykin \cite{malykin06}, some explorers
claim that $\ep$ and $\theta$ have equal signs while some other explorers
claim that $\ep$ and $\theta$ have opposite signs.
Malykin claims that these angles have equal signs while, in contrast, we demonstrate
here convincingly that these angles have opposite signs.
Our demonstration is convincing since it
accompanies a focal identity, \eqref{guliv11}, that
interested explorers can test numerically in order to corroborate our claim
that $\ep$ and $\theta$ have opposite signs.

In order to pave the way to study Thomas precession we present the path from
Einstein velocity addition to the gyroalgebra of gyrogroups and gyrations,
and to the gyrogeometry that coincides with
the hyperbolic geometry of Bolyai and Lobachevsky. We, then,
demonstrate that the concept of Thomas precession in Einstein's
special theory of relativity is a concrete realization of the abstract concept
of gyration in gyroalgebra.

A signed angle $\theta$, $-\pi<\theta<\pi$,
between two non-parallel vectors $\ub$ and $\vb$ in the
Euclidean 3-space $\Rt$ is positive (negative)
if the angle $\theta$ drawn from $\ub$ to $\vb$ is drawn
counterclockwise (clockwise). The relationship between
the Thomas precession signed angle of rotation, $\ep$,
and its generating signed angle, $\theta$, shown in Fig.~\ref{fig115},
is important. Hence, finally, we pay
special attention to the relationship between the
Thomas precession signed angle of rotation and its generating signed angle,
demonstrating that these have opposite signs.

\section{Einstein Velocity Addition and Scalar Multiplication} \label{secc2}

Let $(\Rt,+,\ccdot)$ be the Euclidean $3$-space
with its common vector addition, +, and inner product, $\ccdot$,
and let
\begin{equation} \label{eqcball}
\Rct    = \{\vb\in\Rt: \|\vb\| < c \}
\end{equation}
be the $c$-ball of all relativistically admissible velocities of material
particles, where $c$ is the vacuum speed of light.

Einstein velocity addition is a binary operation, $\op$,
in the $c$-ball $\Rct$ of all relativistically admissible velocities,
given by the equation,
\cite{mybook01},
\cite[Eq.~2.9.2]{urbantkebookeng},\cite[p.~55]{moller52},\cite{fock},
\begin{equation} \label{eq01}
\begin{split}
{\ub}\op{\vb} &=\frac{1}{\unpuvc}
\left\{ {\ub}+ \frac{1}{\gub}\vb+\frac{1}{c^{2}}\frac{\gamma _{{\ub}}}{%
1+\gamma _{{\ub}}}( {\ub}\ccdot{\vb}) {\ub} \right\}
\\[8pt]  &=
\frac{1}{\unpuvc}
\left\{ {\ub}+ {\vb}+\frac{1}{c^{2}}\frac{\gamma _{{\ub}}}{%
1+\gamma _{{\ub}}}\left( {\ub}\times ( {\ub}\times {\vb}%
)\right) \right\}
\end{split}
\end{equation}
for all $\ub,\vb\in\Rct$,
where $\gub$ is the gamma factor given by the equation
\begin{equation} \label{v72gs}
\gvb = \frac{1}{\sqrt{1-\displaystyle\frac{\|\vb\|^2}{c^2}}}
\end{equation}
Here $\ub\ccdot\vb$ and $\|\vb\|$
are the inner product and the norm
in the ball, which the ball $\Rct  $ inherits from its space $\Rt$,
$\|\vb\|^2=\vb\ccdot\vb=\vb^2$.
Recalling that
a nonempty set with a binary operation is called a {\it groupoid},
the Einstein groupoid $(\Rct,\op)$ is called an {\it Einstein gyrogroup}.
A formal definition of the abstract gyrogroup will be presented in
Def.~\ref{defroupx}.

Einstein addition admits scalar multiplication $\od$, giving rise to the
Einstein gyrovector space $(\Rct,\op,\od)$. Remarkably, the resulting
Einstein gyrovector spaces $(\Rcn,\op,\od)$ form the setting for the
Cartesian-Beltrami-Klein ball model of hyperbolic geometry, just as
vector spaces form the setting for the standard Cartesian model of
Euclidean geometry, as we will see in the sequel.

Let $k\od\vb$ be the Einstein addition of $k$ copies of $\vb\in \Rcn$,
that is
$k\od\vb=\vb\op\vb\ldots\op\vb$ ($k$ terms). Then,
\begin{equation} \label{duhnd1}
k\od \vb = c \frac
 {\left(1+\displaystyle\frac{\| \vb \|}{c}\right)^k
- \left(1-\displaystyle\frac{\| \vb \|}{c}\right)^k}
 {\left(1+\displaystyle\frac{\| \vb \|}{c}\right)^k
+ \left(1-\displaystyle\frac{\| \vb \|}{c}\right)^k}
\frac{\vb}{\| \vb\|}
\end{equation}

The definition of scalar multiplication in an Einstein gyrovector space
requires analytically continuing $k$ off the positive integers, thus
obtaining the following definition:
\index{scalar multiplication, Einstein}
\index{Gyrovector space}

\begin{ddefinition}\label{defgvspace}
\index{Gyrovector spaces, def.}
\index{Einstein Scalar Multiplication, def.}
{\bf (Einstein Scalar Multiplication; Einstein Gyrovector Spaces).}
{\it
An Einstein gyrovector space $(\Rsn,\op,\od)$
is an Einstein gyrogroup $(\Rsn,\op)$ with scalar multiplication
$\od$ given by
 \begin{equation} \label{eqmlt03}
 r\od\vb = s \frac
 {\left(1+\displaystyle\frac{\|\vb\|}{s}\right)^r
- \left(1-\displaystyle\frac{\|\vb\|}{s}\right)^r}
 {\left(1+\displaystyle\frac{\|\vb\|}{s}\right)^r
+ \left(1-\displaystyle\frac{\|\vb\|}{s}\right)^r}
\frac{\vb}{\|\vb\|}
= s \tanh( r\,\tanh^{-1}\frac{\|\vb\|}{s})\frac{\vb}{\|\vb\|}
 \end{equation}
where $r$ is any real number, $r\in\Rb$,
$\vb\in\Rsn$, $\vb\ne\bz$, and $r\od \bz=\bz$, and with
which we use the notation $\vb\od  r=r\od \vb$.
}
\end{ddefinition}

In the Newtonian limit of large $c$, $c\rightarrow\infty$, the ball $\Rct   $
expands to the whole of its space $\Rt$, as we see from \eqref{eqcball},
and Einstein addition $\op$ in $\Rct   $
reduces to the ordinary vector addition $+$ in $\Rt$,
as we see from \eqref{eq01} and \eqref{v72gs}.

Einstein addition \eqref{eq01} of relativistically
admissible velocities was introduced by Einstein in his 1905 paper
\cite{einstein05} \cite[p.~141]{einsteinfive}
that founded the special theory of relativity.
One has to remember here that the Euclidean 3-vector algebra was not so
widely known in 1905 and, consequently, was not used by Einstein.
Einstein calculated in \cite{einstein05} the behavior
of the velocity components parallel and orthogonal to the relative
velocity between inertial systems, which is as close as one can get
without vectors to the vectorial version \eqref{eq01} of Einstein addition.

We naturally use the abbreviation
$\ub\om\vb=\ub\op(-\vb)$ for Einstein subtraction, so that,
for instance, $\vb\om\vb = \zerb$,
$\om\vb = \zerb\om\vb=-\vb$ and, in particular,
\begin{equation} \label{eq01a}
\om(\ub\op\vb) = \om\ub\om\vb
\end{equation}
and
\begin{equation} \label{eq01b}
\om\ub\op(\ub\op\vb) = \vb
\end{equation}
for all $\ub,\vb$ in the ball $\Rct$,
in full analogy with vector addition and subtraction in $\Rt$. Identity
\eqref{eq01a} is known as the {\it automorphic inverse property},
and Identity
\eqref{eq01b} is known as the {\it left cancellation law} of Einstein
addition \cite{mybook03}.
We may note that
Einstein addition does not obey the naive right counterpart of the
left cancellation law \eqref{eq01b} since, in general,
\begin{equation} \label{eq01c}
(\ub\op\vb)\om\vb \ne \ub
\end{equation}
The seemingly lack of a right cancellation law for
Einstein addition is repaired in \eqref{rightcanc}
by the introduction of a second binary operation,
called {\it Einstein coaddition}, \eqref{hvudk},
which captures important analogies.

Einstein addition and the gamma factor are related by
the {\it gamma identity},
\begin{subequations} \label{grbsf09}
\begin{equation} \label{grbsf09p0}
\gamma_{\ub\op\vb}^{\phantom{O}}
= \gub\gvb\left( 1 + \frac{\ub\ccdot\vb}{c^2}\right)
\end{equation}
which can be written, equivalently, as
\begin{equation} \label{grbsf09p1}
\gamma_{\om\ub\op\vb}^{\phantom{O}}
= \gub\gvb\left( 1 - \frac{\ub\ccdot\vb}{c^2}\right)
\end{equation}
\end{subequations}
for all $\ub,\vb\in \Rct$.
Here, \eqref{grbsf09p1} is obtained from \eqref{grbsf09p0} by replacing $\ub$ by
$\om \ub=-\ub$.

A frequently used identity that follows immediately from \eqref{v72gs} is
\begin{equation} \label{rugh1ds}
\frac{\vb^2}{c^2} =
\frac{\|\vb\|^2}{c^2} = \frac{\gamma_\vb^2 - 1}{\gamma_\vb^2}
\end{equation}
and, similarly, useful identities that follow immediately from \eqref{grbsf09} are
\begin{subequations} \label{rugh2ds}
\begin{equation} \label{rugh2dsp1}
\frac{\ub\ccdot\vb}{c^2} = -1 + \frac{\gamma_{   \ub\op\vb}^{\phantom{O}}}{\gub\gvb}
\end{equation}
and
\begin{equation} \label{rugh2dsp2}
\frac{\ub\ccdot\vb}{c^2} = 1 - \frac{\gamma_{\om\ub\op\vb}^{\phantom{O}}}{\gub\gvb}
\end{equation}
implying
\begin{equation} \label{rugh3}
\gamma_{\ub\op\vb}^{\phantom{O}} - \gub\gvb
=
-\gamma_{\om\ub\op\vb}^{\phantom{O}} + \gub\gvb
\end{equation}
\end{subequations}

In Identity \eqref{rugh3} the left-hand side seems to be more elegant, in form,
than the right-hand side.
Geometrically, however, the right-hand side of this identity
is advantageous over its left-hand side because the gamma factor
$\gamma_{\om\ub\op\vb}^{\phantom{O}}$ that appears on the
right-hand side possesses a geometric interpretation.
It is a geometric interpretation in hyperbolic triangles, called {\it gyrotriangles},
as explained in \cite[Sec.~2.10]{mybook05}.
To be more specific, we recall in Sec.~\ref{specific} relevant results
from \cite{mybook05}.

Einstein addition is noncommutative. Indeed, in general,
\begin{equation} \label{eqyt01} 
\ub\op\vb\ne\vb\op\ub
\end{equation}
$\ub,\vb\in\Rct$. Moreover, Einstein addition is also nonassociative since,
in general,
\begin{equation} \label{eqyt02}
(\ub\op\vb)\op\wb\ne\ub\op(\vb\op\wb)
\end{equation}
$\ub,\vb,\wb\in\Rct   $.

It seems that following the breakdown of
commutativity and associativity in Einstein addition some mathematical
regularity has been lost in the transition from
Newton's velocity vector addition in $\Rt$ to
Einstein's velocity addition \eqref{eq01} in $\Rct   $. This is, however, not the
case since gyrations come to the rescue, as we will see in Sec.~\ref{secc3}.
Owing to the presence of gyrations, the Einstein groupoid $(\Rct,\op)$
has a grouplike structure \cite{grouplike} that we naturally call an
{\it Einstein gyrogroup} \cite{mybook01}.
The formal definition of the resulting abstract gyrogroup will be presented in
Def.~\ref{defroupx}, Sec.~\ref{secc4}.

\section{Linking Einstein Addition to Hyperbolic Geometry}\label{secfm}

The Einstein gyrodistance function, $d(\ub,\vb)$
in an Einstein gyrovector space $(\Rcn,\op,\od)$ is given by
the equation
\begin{equation} \label{eqeindist}
d(\ub,\vb)=\|\ub\om\vb\|
\end{equation}
$\ub,\vb\in \Rcn$.
We call it a {\it gyrodistance function}
in order to emphasize the analogies it shares
with its Euclidean counterpart, the distance function
$\|\ub-\vb\|$ in $\Rn$.
Among these analogies is the gyrotriangle inequality according to which
\begin{equation} \label{rib01}
\|\ub\op\vb\| \le \|\ub\| \op \|\vb\|
\end{equation}
for all $\ub,\vb \in\Rcn$.
For this and other analogies that distance and gyrodistance functions share
see \cite{mybook02, mybook03}.

In a two dimensional Einstein gyrovector space $(\Rctwo,\op,\od)$
the squared gyrodistance between a
point $\xb\in \Rctwo$ and an infinitesimally
nearby point $\xb+d\xb\in \Rctwo$, $d\xb=(dx_1,~dx_2)$,
is defined by the equation \cite[Sec.~7.5]{mybook03} \cite[Sec.~7.5]{mybook02}
\begin{equation} \label{eqwdfs00}
\begin{split}
ds^2 &= \|(\xb+d\xb)\om\xb\|^2 \\[6pt]
&= Edx_1^2 + 2Fdx_1dx_2 + Gdx_2^2 + \dots
\end{split}
\end{equation}
where, if we use the notation $r^2 = x_1^2+x_2^2$, we have
 \begin{equation} \label{eqwdfs02}
 \begin{split}
E &= c^2\frac{c^2-x_2^2}{(c^2-r^2)^2}  \\[8pt]
F &= c^2\frac{x_1 x_2}{(c^2-r^2)^2}  \\[8pt]
G &= c^2\frac{c^2-x_1^2}{(c^2-r^2)^2}
 \end{split}
 \end{equation}

The triple $(g_{11},g_{12},g_{22})=(E,F,G)$
along with $g_{21}=g_{12}$ is known in
differential geometry as the metric tensor $g_{ij}$ \cite{kreyszig91}.
It turns out to be the metric tensor of the Beltrami-Klein disc model
of hyperbolic geometry \cite[p.~220]{mccleary94}.
Hence, $ds^2$
in \eqref{eqwdfs00}\,--\,\eqref{eqwdfs02}
is the Riemannian line element of the Beltrami-Klein disc model
of hyperbolic geometry, linked to
Einstein velocity addition \eqref{eq01}
and to Einstein gyrodistance function \eqref{eqeindist}
\cite{ungardiff05}.

The link between Einstein gyrovector spaces and the Beltrami-Klein ball model
of hyperbolic geometry, already noted by Fock \cite[p.~39]{fock},
has thus been established in \eqref{eqeindist}\,--\,\eqref{eqwdfs02} in two dimensions.
The extension of the link to higher dimensions is presented in
\cite[Sec.~9, Chap.~3]{mybook01}, \cite[Sec.~7.5]{mybook03} \cite[Sec.~7.5]{mybook02} and
\cite{ungardiff05}.
For a brief account of the history of linking Einstein's velocity
addition law with hyperbolic geometry see
\cite[p.~943]{rhodes04}.

\section{Gyrotriangle, the Hyperbolic Triangle} \label{specific}

In this inspirational section we present recollections from \cite{mybook05}
that intend to motivate Thomas precession explorers
to study the gyrostructure of Einstein addition and its underlying
hyperbolic geometry.

To be specific about the advantage in geometric interpretation of
the use of $\gamma_{\om\ub\op\vb}^{\phantom{O}}$ over the use of
$\gamma_{\ub\op\vb}^{\phantom{O}}$ that appear in Identity \eqref{rugh3}
we recall from \cite{mybook05} that Einstein addition
admits scalar multiplication, $\od$,
turning the Einstein {\it gyrogroup} $(\Rct,\op)$ into a corresponding
Einstein {\it gyrovector space} $(\Rct,\op,\od)$, as we will see in
Sec.~\ref{secfm}.

  
\begin{figure}[t]  
 \centering         
\psfrag{O}[]{$\phantom{O}$}
\psfrag{A1}[]{$U$}
\psfrag{A2}[]{$V$}
\psfrag{A3}[]{$W$}
\psfrag{pa1}{$\gub$}
\psfrag{pa2}[]{$\gvb$}
\psfrag{pa3}{$\gwb=\gmupvb$}
\psfrag{text1}[]{$\ub=\om W \op V$}
\psfrag{text2}[]{$\vb=\om W \op U$}
\psfrag{text3}[]{$\wb=\om U \op V$}
\psfrag{al1}[]{$\alpha$}
\psfrag{al2}[]{$\beta$}
\psfrag{al3}[]{$\gamma$}
\psfrag{formula01}{$w=\|\wb\|=\|\om U \op V\|$}
\psfrag{formula02}{$v=\|\vb\|=\|\om W \op U\|$}
\psfrag{formula03}{$u=\|\ub\|=\|\om W \op V\|$}
\psfrag{fig171ec3}{$\delta=\pi-(\alpha+\beta+\gamma)$}
\psfrag{----chets1}[]{\lower-1.2ex \hbox {$\blacktriangleright$}}
\psfrag{----chets2}[]{\lower-1.26ex \hbox {$\blacktriangleright$}}
\psfrag{----chets3}[]{\lower-1.2ex \hbox {$\blacktriangleright$}}
 \includegraphics[width=10cm]{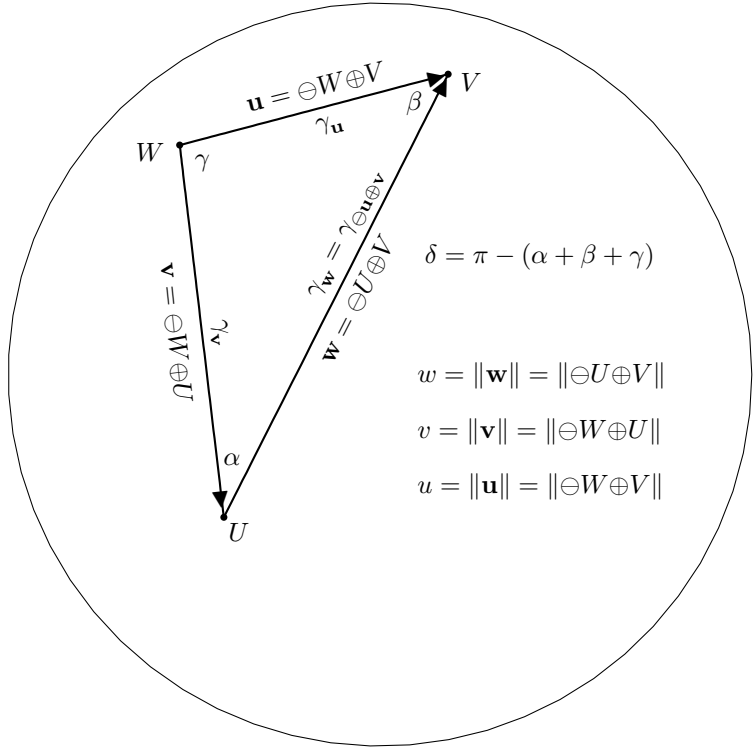}
\caption{
The gyrotriangle $UVW$ in an
Einstein gyrovector space $(\Rsn,\op,\od)$ is shown for $n=2$.
Its sides are presented graphically as gyrosegments that join the vertices.
They form the gyrovectors $\ub,\vb,\wb$,
side-gyrolengths, $u,v,w$,
and gyroangles, $\alpha,\beta,\gamma$.
The gyrotriangle gyroangle sum is less than $\pi$,
the difference, $\delta=\pi-(\alpha + \beta + \gamma)$,
being the gyrotriangular defect.
\label{fig244ein0bm}}
\end{figure}

Let $U$, $V$ and $W$ be the three vertices of a gyrotriangle $UVW$
in an Einstein gyrovector space $(\Rct,\op,\od)$,
shown in Fig.~\ref{fig244ein0bm}.
Then, in full analogy with Euclidean geometry,
the three sides of the gyrotriangle form the three {\it gyrovectors}
\begin{equation} \label{hdke1}
\begin{split}
\ub &= \om W\op V \\
\vb &= \om W\op U \\
\wb &= \om U\op V
\end{split}
\end{equation}
and the corresponding three side-gyrolengths of the gyrotriangle are
\begin{equation} \label{hdke2}
\begin{split}
u &= \|\ub\| = \|\om W\op V\| \\
v &= \|\vb\| = \|\om W\op U\| \\
w &= \|\wb\| = \|\om U\op V\| = \|\om\ub\op\vb\|
\end{split}
\end{equation}
and the side gamma factors of the gyrotriangle are, accordingly,
$\gub$, $\gvb$ and
\begin{equation} \label{hdke3}
\gwb = \gmupvb
\end{equation}
Hence, by \eqref{hdke3}, the gamma factors $\gub$, $\gvb$ and $\gmupvb$
in \eqref{rugh3} can be interpreted geometrically as the gamma factors
of the three sides of a gyrotriangle $UVW$.

For later reference we recall
here that the gyrotriangular defect $\delta$ of the gyrotriangle $UVW$
is given in terms of the gyrotriangle side gamma factors
by the equation
\cite[Theorem 2.32]{mybook06}
\cite[Theorem 6.11]{mybook05}
\begin{equation} \label{hdke4}
\tan\tfrac{\delta}{2} =
\frac{\sqrt{1 + 2\gub\gvb\gwb-\gubs-\gvbs-\gwbs}}
{1+\gub+\gvb+\gwb}
\end{equation}
where $\gwb=\gmupvb$.

It is the gamma identity \eqref{grbsf09p0} that
signaled the emergence of hyperbolic geometry in special relativity
when it was first studied by
Sommerfeld \cite{sommerfeld09} and Vari\v cak \cite{varicak08,varicak10}.
Historically,
it formed the first link between special relativity and the
hyperbolic geometry of Bolyai and Lobachevsky, recently leading to the novel
trigonometry in hyperbolic geometry that became known as {\it gyrotrigonometry},
studied in
\cite{mybook01,mybook02,mybook03,mybook04,mybook06,mybook05}.

\section{The Gyrostructure of Einstein Addition}
\label{secc3}

Vector addition, +, in $\Rt$ is both commutative and associative, satisfying
\begin{alignat}{2}\label{laws11}
\notag
\ub + \vb &= \vb + \ub &&\hspace{1.2cm}\text{Commutative Law}\\
\notag
\ub + (\vb + \wb) &= (\ub + \vb) + \wb &&\hspace{1.2cm}\text{Associative Law}\\
\end{alignat}
for all $\ub,\vb,\wb\in \Rt$.
In contrast, Einstein addition, $\op$, in $\Rct$ is neither commutative nor associative.

In order to measure the extent to which Einstein addition deviates from associativity
we introduce {\it gyrations},
which are maps that are {\it trivial} in the special
cases when the application of $\op$ is associative.
For any $\ub,\vb\in\Rct$ the gyration $\gyruvb$ is an automorphism of the
Einstein groupoid $(\Rct,\op)$ onto itself, given in terms of
Einstein addition by the equation
\begin{equation} \label{jdhbw}
\gyruvb\wb = \om(\ub\op\vb)\op\{\ub\op(\vb\op\wb)\}
\end{equation}
for all $\ub,\vb,\wb\in\Rct$.

We recall that an automorphism of a groupoid $(S,\op)$ is a one-to-one map
$f$ of $S$ onto itself that respects the binary operation, that is,
$f(a\op b)=f(a)\op f(b)$ for all $a,b\in S$.
The set of all automorphisms of a groupoid $(S,\op)$ forms a group,
denoted $\Aut(S,\op)$, where the group operation is given by
automorphism composition.
To emphasize that the gyrations of an Einstein gyrogroup $(\Rct,\op)$
are automorphisms of the gyrogroup, gyrations are also called
{\it gyroautomorphisms}.

A gyration $\gyruvb$, $\ub,\vb\in\Rct$, is {\it trivial} if $\gyruvb\wb=\wb$
for all $\wb\in\Rct$. Thus, for instance, the gyrations $\gyr[\zerb,\vb]$, $\gyr[\vb,\vb]$
and $\gyr[\vb,\om\vb]$ are trivial for all $\vb\in\Rct$, as we see from
\eqref{jdhbw} and \eqref{eq01b}.

Einstein gyrations, which possess their own rich structure,
measure the extent to which Einstein addition deviates from both
commutativity and associativity as we see from the gyrocommutative and the
gyroassociative laws of Einstein addition in the following list of identities
\cite{mybook01,mybook02,mybook03,mybook04,mybook06,mybook05},
each of which has a name:
\begin{alignat}{2}\label{laws00}
\notag
\ub\op\vb &= \gyruvb(\vb\op\ub) &&\hspace{1.2cm}\text{Gyrocommutative Law}\\
\notag
\ub\op(\vb\op \wb) &= (\ub\op \vb)\op\gyruvb \wb &&\hspace{1.2cm}\text{Left Gyroassociative Law}\\
\notag
(\ub\op \vb)\op \wb &=\ub\op(\vb\op\gyrvub \wb) &&\hspace{1.2cm}\text{Right Gyroassociative Law}\\
\notag
\gyr [\ub\op \vb,\vb] &= \gyruvb &&\hspace{1.2cm}\text{Gyration Left Loop Property}\\
\notag
\gyr [\ub,\vb\op \ub] &= \gyruvb &&\hspace{1.2cm}\text{Gyration Right Loop Property}\\
\notag
\gyr[\om\ub,\om\vb] &=\gyr [\ub,\vb] &&\hspace{1.2cm}\text{Gyration Even Property} \\
\notag
(\gyr [\ub,\vb])^{-1} &=\gyr [\vb,\ub] &&\hspace{1.2cm}\text{Gyration Inversion Law} \\
\end{alignat}
for all $\ub,\vb,\wb\in \Rct$.

Einstein addition is thus regulated by gyrations to which it gives rise owing
to its nonassociativity, so that
Einstein addition and its gyrations are inextricably linked.
The resulting gyrocommutative gyrogroup structure of Einstein addition
was discovered in 1988 \cite{parametrization}.
Interestingly, gyrations are the mathematical abstraction of the
relativistic mechanical effect known as {\it Thomas precession} \cite[Sec.~10.3]{mybook03},
as we will see in Sec.~\ref{secgm}.

\section{Gyrations}\label{secgyrations73}

Owing to its nonassociativity,
Einstein addition gives rise in \eqref{jdhbw} to nontrivial gyrations
\begin{equation} \label{hkedu}
\gyruvb : \Rct~\rightarrow~\Rct
\end{equation}
for any $\ub,\vb\in\Rct$.
Gyrations, in turn, regulate Einstein addition, endowing it with
the rich structure of a gyrocommutative gyrogroup that will be
formalize in Sec.~\ref{secc4}.

The gyration equation is expressed in \eqref{jdhbw}
in terms of Einstein addition. Expressing it explicitly,
in terms of vector addition and vector scalar product rather than Einstein addition,
we obtain the equation
\begin{equation} \label{hdge1ein}
\gyruvb\wb = \wb + \frac{A\ub+B\vb}{D}
\end{equation}
where
\begin{equation} \label{hdgej2ein}
\begin{split}
 A &=-\frac{1}{c^2}\frac{\gubs}{(\gub+1)} (\gvb-1) (\ub\ccdot\wb)
 +
 \frac{1}{c^2}\gub\gvb (\vb\ccdot\wb)
\\[8pt] & \phantom{=} ~+
 \frac{2}{c^4} \frac{\gubs\gvbs}{(\gub+1)(\gvb+1)} (\ub\ccdot\vb) (\vb\ccdot\wb)
\\[8pt]
B &=- \frac{1}{c^2}
\frac{\gvb}{\gvb+1}
\{\gub(\gvb+1)(\ub\ccdot\wb) + (\gub-1)\gvb(\vb\ccdot\wb) \}
\\[8pt]
D &= \gub\gvb(1+ \frac{\ub\ccdot\vb}{c^2}) +1 = \gamma_{\ub\op\vb}^{\phantom{O}} + 1 > 1
\end{split}
\end{equation}
for all $\ub,\vb,\wb\in\Rct$.
Clearly, the domain of $\ub$ and $\vb$ in \eqref{hdgej2ein} must be restricted
to the open ball $\Rct$ to insure the reality of the
Lorentz factors $\gub$ and $\gvb$.
In contrast, however, the domain of $\wb$ need not be restricted to the ball.

Allowing $\wb\in\Rt\supset\Rct$ in \eqref{hdge1ein}\,--\,\eqref{hdgej2ein},
that is, extending the domain of $\wb$ from $\Rct$ to $\Rt$,
gyrations $\gyr[\ub,\vb]$
are expendable from self-maps of $\Rct$ to linear self-maps of $\Rt$ for all $\ub,\vb\in\Rct$.
Indeed,
\begin{equation} \label{hekdn}
\gyruvb(\ro\wb_1+\rt\wb_2) = \ro\gyruvb\wb_1+\rt\gyruvb\wb_2
\end{equation}
for all $\ub,\vb\in\Rct$, $\wb\in\Rt$ and $\ro,\rt\in\Rb$.

In each of the three special cases when
(i) $\ub=\zerb$, or
(ii) $\vb=\zerb$, or
(iii) $\ub$ and $\vb$ are parallel
in $\Rt$, $\ub\|\vb$, we have
$A\ub+B\vb=\zerb$, so that in these cases $\gyr[\ub,\vb]$ is trivial. Thus, we have
\begin{equation} \label{sprdhein}
\begin{split}
\gyr[\zerb,\vb]\wb &= \wb  \\
\gyr[\ub,\zerb]\wb &= \wb \\
\gyr[\ub,\vb]\wb &= \wb , \hspace{1.2cm} \ub\|\vb
\end{split}
\end{equation}
for all $\ub,\vb\in\Rct\subset\Rt$ and all $\wb\in\Rt$.

It follows from \eqref{hdge1ein}\,--\,\eqref{hdgej2ein} that
\begin{equation} \label{eq1ffmznein}
\gyr[\vb,\ub](\gyr[\ub,\vb]\wb) = \wb
\end{equation}
for all $\ub,\vb\in\Rct$, $\wb\in\Rt$, so that gyrations are invertible
linear maps of $\Rt$, the inverse, $\gyr^{-1}[\ub,\vb]$, \eqref{laws00},
of $\gyr[\ub,\vb]$ being $\gyr[\vb,\ub]$.
We thus obtain from \eqref{eq1ffmznein} the gyration inversion property in \eqref{laws00},
\begin{equation} \label{feknd}
\gyr^{-1}[\ub,\vb] = \gyr[\vb,\ub]
\end{equation}
for all $\ub,\vb\in\Rct$.

Gyrations keep the inner product of
elements of the ball $\Rct$ invariant, that is,
\begin{equation} \label{eq005}
\gyruvb\ab\ccdot\gyruvb\bb = \ab\ccdot\bb
\end{equation}
for all $\ab,\bb,\ub,\vb\in \Rct$. Hence, in particular, $\gyruvb$ is an
{\it isometry} of $\Rct$,
keeping the norm of elements of the ball $\Rct$ invariant,
\begin{equation} \label{eq005a}
\|\gyruvb \wb\| = \|\wb\|
\end{equation}
Accordingly, for any $\ub,\vb\in \Rct$,
$\gyruvb$ represents a rotation of the ball $\Rct$ about its origin.

The invertible self-map $\gyruvb$ of $\Rct$ respects Einstein addition in $\Rct$,
\begin{equation} \label{eq005b}
\gyruvb (\ab \op \bb) = \gyruvb\ab \op \gyruvb\bb
\end{equation}
for all $\ab,\bb,\ub,\vb\in \Rct$,
so that, by \eqref{feknd} and \eqref{eq005b},
$\gyruvb$ is an automorphism of the Einstein groupoid $(\Rct,\op)$.

\section{Gyrogroups}\label{secc4}

Taking the key features of the Einstein groupoid $(\Rct,\op)$ as
axioms, and guided by analogies with groups,
we are led to the formal gyrogroup definition in which
gyrogroups turn out to form a most natural generalization of groups.
Definitions related to groups and gyrogroups thus follow.

\begin{ddefinition}\label{defgroup01}
{\bf (Groups).}
{\it
A groupoid $(G,\,+)$ is a group if its binary operation satisfies the following axioms.
In $G$ there is at least one element, 0, called a left identity, satisfying

\noindent
(G1) \hspace{1.2cm} 0+a=a

\noindent
for all $a\in G$. There is an element $0\in G$ satisfying axiom $(G1)$ such
that for each $a\in G$ there is an element $-a\in G$, called a left inverse of $a$,
satisfying

\noindent
(G2) \hspace{1.2cm} $-a+a=0$

\noindent
Moreover, the binary operation obeys the associative law

\noindent
(G3) \hspace{1.2cm} $(a+b)+c = a+(b+c)$

\noindent
for all $a,b,c\in G$.
}
\end{ddefinition}

Groups are classified into commutative and noncommutative groups.

\begin{ddefinition}\label{defgroup02}
{\bf (Commutative Groups).}
{\it
A group $(G,\,+)$ is commutative if
its binary operation obeys the commutative law

\noindent
(G6) \hspace{1.2cm} $a+b=b+a$

\noindent
for all $a,b\in G$.
}
\end{ddefinition}

\begin{ddefinition}\label{defroupx}
{\bf (Gyrogroups).}
{\it
A groupoid $(G , \op )$
is a gyrogroup if its binary operation satisfies the following axioms.
In $G$ there is at least one element, $0$, called a left identity, satisfying

\noindent
(G1) \hspace{1.2cm} $0 \op a=a$

\noindent
for all $a \in G$. There is an element $0 \in G$ satisfying axiom $(G1)$ such
that for each $a\in G$ there is an element $\om a\in G$, called a
left inverse of $a$, satisfying

\noindent
(G2) \hspace{1.2cm} $\om a \op a=0\,.$

\noindent
Moreover, for any $a,b,c\in G$ there exists a unique element $\gyr[a,b]c \in G$
such that the binary operation obeys the left gyroassociative law

\noindent
(G3) \hspace{1.2cm} $a\op(b\op c)=(a\op b)\op\gyrab c\,.$

\noindent
The map $\gyr[a,b]:G\to G$ given by $c\mapsto \gyr[a,b]c$
is an automorphism of the groupoid $(G,\op)$, that is,

\noindent
(G4) \hspace{1.2cm} $\gyrab\in\Aut (G,\op) \,,$

\noindent
and the automorphism $\gyr[a,b]$ of $G$ is called
the gyroautomorphism, or the gyration, of $G$ generated by $a,b \in G$.
The operator $\gyr : G\times G\rightarrow\Aut (G,\op)$ is called the
gyrator of $G$.
Finally, the gyroautomorphism $\gyr[a,b]$ generated by any $a,b \in G$
possesses the left loop property

\noindent
(G5) \hspace{1.2cm} $\gyrab=\gyr [a\op b,b] \,.$
}
\end{ddefinition}

The gyrogroup axioms ($G1$)\,--\,($G5$)
in Def.~\ref{defroupx} are classified into three classes:
\begin{enumerate}
\item
The first pair of axioms, $(G1)$ and $(G2)$ in Def.~\ref{defroupx},
is a reminiscent of the group axioms $(G1)$ and $(G2)$ in Def.~\ref{defgroup01}.
\item
The last pair of axioms, $(G4)$ and $(G5)$ in Def.~\ref{defroupx}, presents the gyrator
axioms.
\item
The middle axiom, $(G3)$ in Def.~\ref{defroupx},
is a hybrid axiom linking the two pairs of axioms in items (1) and (2).
\end{enumerate}

As in group theory, we use the notation
$a \om b = a \op (\om b)$
in gyrogroup theory as well.

In full analogy with groups, gyrogroups are classified into gyrocommutative and
non-gyrocommutative gyrogroups.

\begin{ddefinition}\label{defgyrocomm}
{\bf (Gyrocommutative Gyrogroups).}
{\it
A gyrogroup $(G, \oplus )$ is gyrocommutative if
its binary operation obeys the gyrocommutative law

\noindent
(G6) \hspace{1.2cm} $a\oplus b=\gyrab(b\oplus a)$

\noindent
for all $a,b\in G$.
}
\end{ddefinition}

While it is clear how to define a right identity and a right inverse
in a gyrogroup,
the existence of such elements is not presumed. Indeed, the existence
of a unique identity and a unique inverse, both left and right, is a
consequence of the gyrogroup axioms, as the following theorem shows, along with
other immediate results.

\begin{ttheorem}\label{thm2d1}
{\bf (First Gyrogroup Properties).}
Let $(G,\, \op )$ be a gyrogroup. For any elements
$a,b,c,x\in G$ we have the following results:
\begin{enumerate}
\item \label{qwhdb01}
If $a\op b=a\op c$, then $b=c$ (general left cancellation law;
see item (9) below).
\item \label{qwhdb02}
$\gyr [0,a]=I$ for any left identity $0$ in $G$.
\item \label{qwhdb03}
$\gyr [x,a]=I$ for any left inverse $x$ of $a$ in $G$.
\item \label{qwhdb04}
$\gyr [a,a]=I$
\item \label{qwhdb05}
There is a left identity which is a right identity.
\item \label{qwhdb06}
There is only one left identity.
\item \label{qwhdb07}
Every left inverse is a right inverse.
\item \label{qwhdb08}
There is only one left inverse, $\om a$, of $a$, and $\om (\om a)=a$.
\item \label{qwhdb09}
The Left Cancellation Law:
\begin{equation} \label{leftcanc7}
\om a\op (a\op b)=b
\end{equation}
\item \label{qwhdb10}
The Gyrator Identity:
\begin{equation} \label{gyrator7}
\gyrab x = \om (a\op b)\op  \{ a\op (b\op x) \}
\end{equation}
\item \label{qwhdb11}
$\gyrab 0 = 0\,.$
\item \label{qwhdb12}
$\gyrab (\om x) = \om \gyrab x\,.$
\item \label{qwhdb13}
$\gyr [a,0]=I\,.$
\end{enumerate}
\end{ttheorem}
\begin{proof}
\noindent
\begin{enumerate}
\item \label{qwhdc01}
Let $x$ be a left inverse of $a$ corresponding to a left identity, $0$, in
$G$. We have $x \op  (a \op  b)$ = $x \op  (a \op  c)$, implying
$(x \op  a) \op \gyr [x,a]b$ = $(x \op  a) \op \gyr [x,a]c$
by left gyroassociativity.
Since $0$ is a left identity,
$\gyr [x,a]b=\gyr [x,a]c$.
Since automorphisms are bijective, $b=c$.
\item \label{qwhdc02}
By left gyroassociativity we have for any left identity $0$ of $G$,
$a\op x$ = $0\op (a\op x)$ = $(0\op a)\op \gyr [0,a]x$ = $a\op \gyr [0,a]x$.
Hence, by item \ref{qwhdc01} above we have
$x=\gyr [0,a]x$ for all $x\in G$ so that $\gyr [0,a]=I$.
\item \label{qwhdc03}
By the left loop property and by item \ref{qwhdc02} above we have
$\gyr [x,a]=\gyr [x\op a,a]=\gyr [0,a]=I$.
\item \label{qwhdc04}
Follows from an application of the left loop property and item \ref{qwhdc02} above.
\item \label{qwhdc05}
Let $x$ be a left inverse of $a$ corresponding to a left identity, $0$, of $G$.
Then by left gyroassociativity and item \ref{qwhdc03} above,
$x\op (a\op 0)$ = $(x\op a)\op \gyr [x,a]0=0\op 0=0=x\op a$.
Hence, by (1), $a\op 0=a$ for all $a\in G$ so that $0$ is a right
identity.
\item \label{qwhdc06}
Suppose $0$ and $0^*$ are two left identities, one of which, say $0$, is
also a right identity. Then
$0=0^* \op 0=0^*$.
\item \label{qwhdc07}
Let $x$ be a left inverse of $a$. Then
$x\op (a\op x)$ = $(x\op a)\op \gyr [x,a]x$ = $0\op x=x=x\op 0$,
by left gyroassociativity, (G2) of Def.~\ref{defroupx} and items
\ref{qwhdc03}, \ref{qwhdc05}, \ref{qwhdc06} above.
By item \ref{qwhdc01} we have $a\op x=0$ so that $x$ is a right inverse of $a$.
\item \label{qwhdc08}
Suppose $x$ and $y$ are left inverses of $a$. By item \ref{qwhdc07} above, they are also
right inverses, so $a\op x=0=a\op y$. By item \ref{qwhdc01}, $x=y$.
Let $\om a$ be the resulting unique inverse of $a$. Then $\om a\op a=0$ so that the inverse
$\om (\om a)$ of $\om a$ is $a$.
\item \label{qwhdc09}
By left gyroassociativity and by \ref{qwhdc03} we have
\begin{equation} \label{rendk1}
\om a\op (a\op b)=(\om a\op a)\op \gyr [\om a,a]b=b
\end{equation}
\item \label{qwhdc10}
By an application of the left cancellation law in item \ref{qwhdc09}
to the left gyroassociative law (G3) in Def.~\ref{defroupx} we obtain the result
in item \ref{qwhdc10}.
\item \label{qwhdc11}
We obtain item \ref{qwhdc11} from \ref{qwhdc10} with $x=0$.
\item \label{qwhdc12}
Since $\gyrab$ is an automorphism of $(G,\op )$ we have from \ref{qwhdc11}
\begin{equation} \label{rendk2}
\gyrab(\om x)\op \gyrab x=\gyrab(\om x\op x)=\gyrab 0=0
\end{equation}
and hence the result.
\item \label{qwhdc13}
We obtain item \ref{qwhdc13} from \ref{qwhdc10} with $b=0$,
and a left cancellation, item \ref{qwhdc09}.\\[-23pt]
\end{enumerate}
\end{proof}

Einstein addition admits scalar multiplication, giving rise to
gyrovector spaces.
Studies of gyrogroup theory and gyrovector space theory along with
applications in hyperbolic geometry and in
Einstein's special theory of relativity are presented in
\cite{mybook01,mybook02,mybook03,mybook04,mybook06,mybook05}.

We thus see in this section that
Einsteinian velocity addition, $\op$, in $\Rct$ of relativistically admissible velocities
gives rise to the gyrocommutative gyrogroup $(\Rct,\op)$,
just as
Newtonian velocity addition, $+$, in $\Rt$ of classical, Newtonian velocities
gives rise to the commutative group $(\Rt,+)$.
Newtonian velocity addition, in turn, is given by the common vector addition in $\Rt$.
Clearly, it is owing to the presence of nontrivial gyrations that gyrogroups are
generalized groups.
Accordingly, gyrations provide the missing link between
the common vector addition and Einstein velocity addition.

\section{The Euclidean and Hyperbolic Lines}\label{subeuclin}


  
\begin{figure}[t]  
 \centering         
%
\psfrag{a}[]{$\,\,A$}
\psfrag{b}[]{$\! B$}
\psfrag{mab}{$m_{_{A,B}}$}
\psfrag{pab}{$P$}
\psfrag{formula00}[]
{$d(A,P) + d(P,B)=d(A,B)$}
\psfrag{formula01}[]{$\boxed{A + ( -  A +  B) t}$}
\psfrag{formula02}[]{$-\infty\le t \le\infty$}
\psfrag{formula03}{$m_{_{A,B}} = A + ( -  A +  B)\half$}
\psfrag{formula04}{$d(A,B) = \|A -  B\|$}
\psfrag{formula05}{$d(A,m_{_{A,B}}) = d(B,m_{_{A,B}})$}
%
 \includegraphics[width=9cm]{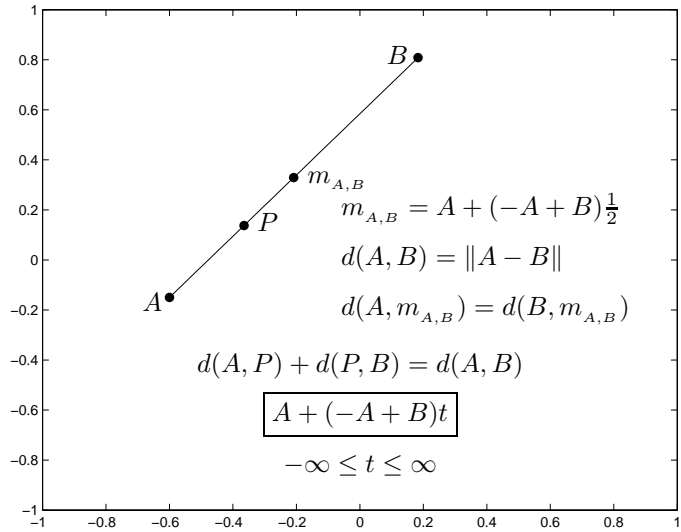}
\caption{
The Cartesian coordinates for the Euclidean plane $\Rtwo$,
$(x_1,x_2)$, $x_1^2 + x_2^2 < \infty$, are shown.
The points $A$ and $B$ in the Euclidean plane $\Rtwo$
are given, with respect to these Cartesian coordinates, by
$A = (-0.60,-0.15)$ and $B = (0.18,0.80)$.
The distance function, $d(A,B)=\|A-B\|$, and the equation of the
line through $A$ and $B$ are shown along with the triangle equality
for a generic point $P$ on the line.
The midpoint $m_{_{A,B}}$ of the segment $AB$ is reached when the line parameter is $t=1/2$.
\label{fig163aceuc3m}}
\end{figure}

As shown in Figs.~\ref{fig163aceuc3m}\,--\,\ref{fig163acein3m},
we introduce Cartesian coordinates into $\Rn$ in the usual way
in order to specify uniquely each point $P$ of the Euclidean $n$-space $\Rn$
by an $n$-tuple of real numbers, called the coordinates, or components, of $P$.
Cartesian coordinates provide a method of indicating the position of points
and rendering graphs on a two-dimensional Euclidean plane $\Rtwo$
and in a three-dimensional Euclidean space $\Rt$.

As an example, Fig.~\ref{fig163aceuc3m} presents a Euclidean plane $\Rtwo$
equipped with a Cartesian coordinate system $\Sigma$.
The position of points $A$ and $B$ and their midpoint $M_{AB}$ with respect to
$\Sigma$ are shown.

The set of all points
\begin{equation} \label{fknced}
A + (-A+B)t
\end{equation}
$t\in\Rb$, forms a Euclidean line. The segment of this line, corresponding to
$1\le t\le1$, and a generic point $P$ on the segment, are shown in
Fig.~\ref{fig163aceuc3m}. Being collinear, the points $A,P$ and $B$ obey the
triangle equality $d(A,P)+d(P,B)=d(A,B)$, where $d(A,B)=\|-A+B\|$ is the
Euclidean distance function in $\Rn$.

Fig.~\ref{fig163aceuc3m} demonstrates the use of the standard Cartesian model
of Euclidean geometry for graphical presentations.
In a fully analogous way, Fig.~\ref{fig163acein3m} demonstrates the use of the
Cartesian-Beltrami-Klein model of hyperbolic geometry; see also
\cite[Figs.~2.3-2.4]{mybook05}.

  
\begin{figure}[t]  
 \centering         
%
\psfrag{a}[]{$\,\,A$}
\psfrag{b}[]{$\! B$}
\psfrag{mab}{$m_{_{A,B}}$}
\psfrag{pab}{$P$}
\psfrag{formula00}[]
{$d(A,P) \op d(P,B)=d(A,B)$}
\psfrag{formula01}[]{$\boxed{A\op(\om A\op B)\od t}$}
\psfrag{formula02}[]{$-\infty\le t \le\infty$}
\psfrag{formula03}{$m_{_{A,B}} = A\op(\om A\op B)\od\half$}
\psfrag{formula04}{$d(A,B) = \|A\om B\|$}
\psfrag{formula05}{$d(A,m_{_{A,B}}) = d(B,m_{_{A,B}})$}
%
 \includegraphics[width=9cm]{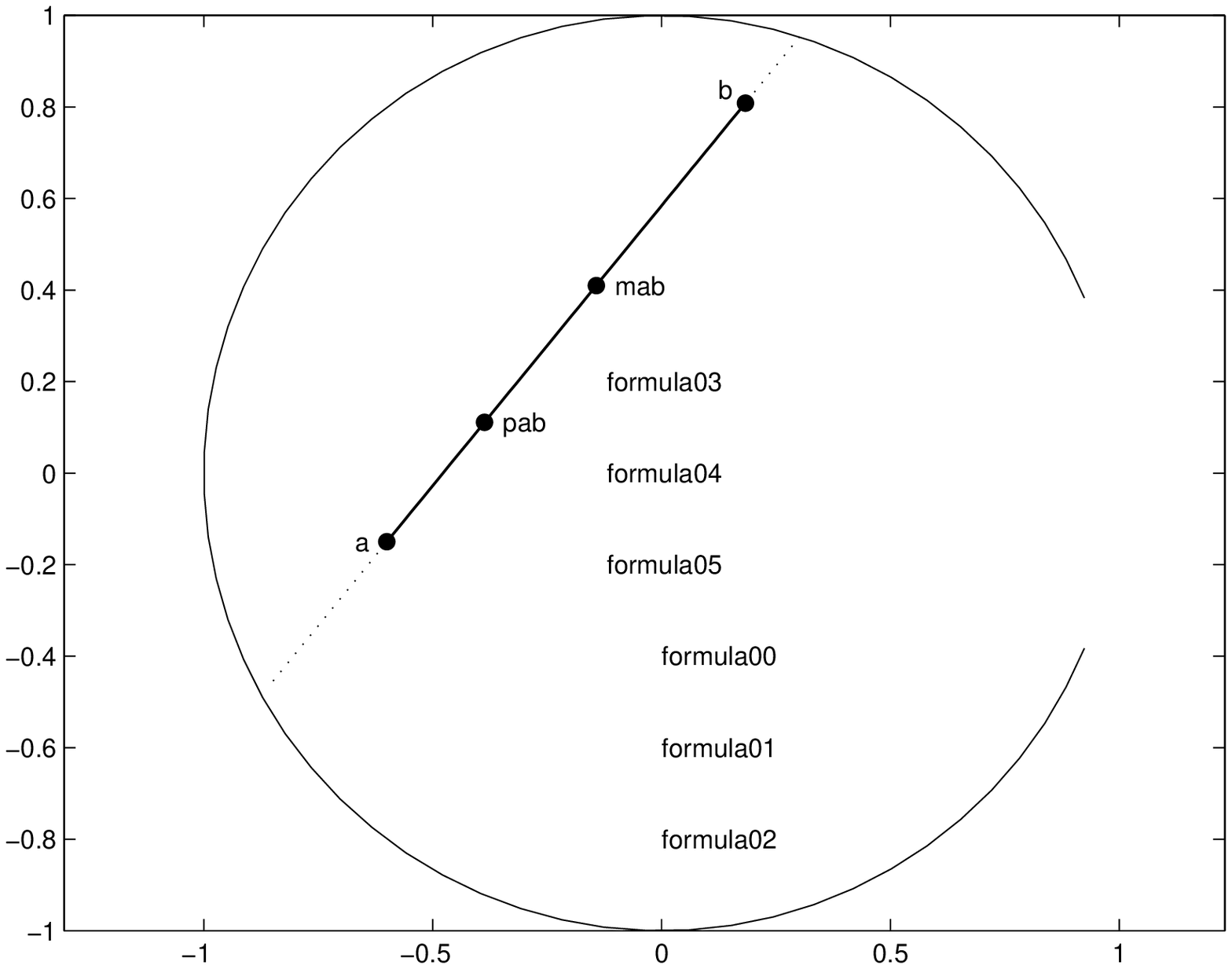}
\caption{
The Cartesian coordinates for the unit disc in the Euclidean plane $\Rtwo$,
$(x_1,x_2)$, $x_1^2 + x_2^2 <1$, are shown.
These, inside the disc, are mow considered the Cartesian coordinates for the
Einstein gyrovector plane $(\Rtwo,\op,\od)$.
The points $A$ and $B$ in the Einstein gyrovector plane $(\Rtwo,\op,\od)$ 
are given, with respect to these Cartesian coordinates, by
$A = (-0.60,-0.15)$ and $B = (0.18,0.80)$.
The gyrodistance function, $d(A,B)=\|A\om B\|$, and the equation of the
gyroline through $A$ and $B$ are shown along with the gyrotriangle equality
for a generic point $P$ on the gyroline.
The gyromidpoint $m_{_{A,B}}$ of the gyrosegment $AB$,
reached when the gyroline parameter is $t=1/2$,
shares obvious analogies with its Euclidean counterpart,
the midpoint in Fig.~\ref{fig163aceuc3m}.
\label{fig163acein3m}}
\end{figure}

Now, let $A,B\in\Rsn$ be two distinct points of the Einstein gyrovector space
$(\Rsn,\op,\od)$, and let $t\in\Rb$ be a real parameter. Then,
in full analogy with the Euclidean line \eqref{fknced},
the graph of the set of all points, Fig.~\ref{fig163acein3m},
\begin{equation} \label{eqcurve94}
A\op(\om A\op B)\od t
\end{equation}
$t\in\Rb$,
in the Einstein gyrovector space $(\Rsn,\op,\od)$ is a chord of the ball $\Rsn$.
As such, it is a geodesic line of the
Cartesian-Beltrami-Klein ball model of hyperbolic geometry.
The geodesic line \eqref{eqcurve94} is the unique geodesic passing through the
points $A$ and $B$. It passes through the point $A$
when $t=0$ and, owing to the left cancellation law, \eqref{leftcanc7},
it passes through the point $B$ when $t=1$.
Furthermore, it passes through the midpoint $M_{ {A,B}}$ of $A$ and $B$
when $t=1/2$.
Accordingly, the {\it gyrosegment} that joins the points
$A$ and $B$ in
Fig.~\ref{fig163acein3m}
is obtained from gyroline \eqref{eqcurve94} with $0\le t\le1$.

\section{Thomas Precession}\label{secgm}

It is owing to the gyrocommutative law, \eqref{laws00}, of Einstein addition that
Thomas precession of Einstein's special theory of relativity
is recognized as a concrete example of the abstract gyrogroup gyration
in Def.~\ref{defroupx}.
Accordingly, the gyrogroup gyration is an extension by abstraction
of the relativistic mechanical effect known as Thomas precession.

The gyrocommutative law of Einstein velocity addition
was already known to Silberstein in 1914 \cite{silberstein14} in the following sense:
According to his 1914 book,
Silberstein knew that the Thomas precession generated by $\ub,\vb\in\Rct$
is the unique rotation
that takes $\vb\op\ub$ into $\ub\op\vb$ about an axis perpendicular to the
plane of $\ub$ and $\vb$ through an angle $< \pi$ in $\Rt$,
thus giving rise to the gyrocommutative law.
However, obviously, Silberstein did not use the terms
``Thomas precession'' and ``gyrocommutative law''.
These terms have been coined
later, respectively,
(i) following Thomas' 1926 paper \cite{thomas26},
and (ii) in 1991 \cite{grouplike,gaxioms}, following the discovery of the
accompanying gyroassociative law of Einstein addition
in 1988 \cite{parametrization}.

A description of the 3-space rotation, which since 1926 is named
after Thomas,
is found in Silberstein's 1914 book \cite{silberstein14}.
In 1914 Thomas precession did not have a name, and
Silberstein called it in his 1914 book a ``certain space-rotation''
\cite[p.~169]{silberstein14}.
An early study of Thomas precession, made by the famous
mathematician Emile Borel in 1913, is described in his 1914 book
\cite{borel14} and, more recently, in \cite{stachel95}.
According to Belloni and Reina \cite{belloni86},
Sommerfeld's
route to Thomas precession dates back to 1909.
However, prior to Thomas discovery the relativistic peculiar
3-space rotation had a most uncertain physical status
\cite[p.~119]{walter99b}.
The only knowledge Thomas had in 1925 about the peculiar relativistic
gyroscopic precession \cite{jonson07},
however, came from De Sitter's formula describing the
relativistic corrections for the motion of the moon, found in
Eddington's book \cite{eddington24}, which was just published at that time
\cite[Sec.~1, Chap.~1]{mybook01}.

The physical significance of the peculiar rotation in
special relativity emerged in 1925 when Thomas relativistically re-computed
the precessional frequency of the doublet separation in the fine structure
of the atom, and thus rectified a missing factor of 1/2. This correction has
come to be known as the
{\it Thomas half} \cite{chrysos06}, presented in \eqref{kurni109}.
Thomas' discovery of the relativistic precession of the
electron spin on Christmas 1925 thus led to the understanding of the
significance of the relativistic
effect which became known as {\it Thomas precession}.
Llewellyn Hilleth Thomas died in
Raleigh, NC, on April 20, 1992.
A paper \cite{bloch02} dedicated to the centenary of the birth of
Llewellyn H. Thomas (1902\,--\,1992) describes the
Bloch gyrovector of quantum information and computation.

Once recognized as gyration, it is clear that Thomas precession
owes its existence solely to the nonassociativity of
Einstein addition of Einsteinian velocities. Accordingly, Thomas precession has
no classical counterpart since the addition of classical,
Newtonian velocities is associative.

It is widely believed that special relativistic effects are negligible when
the velocities involved are much less than the vacuum speed of light $c$.
Yet, Thomas precession effect in the orbital motion of spinning
electrons in atoms
is clearly observed in resulting spectral lines despite the
speed of electrons in atoms being small compared with the speed of light.
One may, therefore, ask whether it is possible to furnish a classical
background to Thomas precession
\cite{keown}.
Hence, it is important to realize that Thomas precession stems from
the nonassociativity of Einsteinian velocity addition, so that it has no echo in
Newtonian velocities.

In 1966, Ehlers, Rindler and Robinson \cite{rindler66}
proposed a new formalism for dealing with the Lorentz group. Their formalism,
however, did not find its way to the mainstream literature. Therefore,
thirty three years later, two of them suggested considering the
``notorious Thomas precession formula''
(in their words \cite[p.~431]{rindler99})
as an indicator of the quality of a formalism for dealing with the Lorentz
group. The idea of Rindler and Robinson to use the
``notorious Thomas precession formula''
as an indicator works fine in the analytic hyperbolic geometric viewpoint
of special relativity, where the
ugly duckling of special relativity, the
``notorious Thomas precession formula'',
becomes the beautiful swan, the ``gyrostructure'', of special relativity
and its underlying analytic hyperbolic geometry.
The abstract Thomas precession, called gyration, is now recognized as the
missing link between
classical mechanics with its underlying Euclidean geometry
and
relativistic mechanics with its underlying hyperbolic geometry.

\section{Thomas Precession Matrix}\label{secfv}

For any two vectors $\ab,\bb\in\Rt$, $\ab=(a_1,a_2,a_3)$, {\it etc.},
determined by their components
with respect to a given Cartesian coordinate system, we define the
square $3 \times 3$ matrix $\Omega(\ab,\bb)$ by the equation
\begin{equation}
\Omega (\ab,\bb) = -
\begin{pmatrix}
                   a_1 b_1 & a_1 b_2 & a_1 b_3 \\
                   a_2 b_1 & a_2 b_2 & a_2 b_3 \\
                   a_3 b_1 & a_3 b_2 & a_3 b_3
\end{pmatrix}
                    +
\begin{pmatrix}
                   a_1 b_1 & a_2 b_1 & a_3 b_1 \\
                   a_1 b_2 & a_2 b_2 & a_3 b_2 \\
                   a_1 b_3 & a_2 b_3 & a_3 b_3
\end{pmatrix}
\end{equation}
or, equivalently,
\begin{equation}
\Omega (\ab,\bb ) = -
\begin{pmatrix}
                 0     &  \omega_3 & -\omega_2  \\
             -\omega_3 &     0     &  \omega_1  \\
              \omega_2 & -\omega_1 &     0      \\
\end{pmatrix}
\end{equation}
where
\begin{equation}
\wb = (\omega_1 ,\, \omega_2 ,\, \omega_3 ,\, ) = \ab \times \bb
\end{equation}
Accordingly,
\begin{equation}
\Omega(\ab,\bb)\xb = (\ab\times\bb)\times\xb = -\ab(\bb\ccdot\xb)+\bb(\ab\ccdot\xb)
\end{equation}
for any $\xb\in\Rt$. Hence,
\begin{enumerate}
\item
$\Omega(\ab,\bb)=0$ if and only if
$\ab\times\bb=\bz$;
\item
and
\begin{equation}\label{omegazero}
\Omega(\ab,\bb)(\ab\times\bb) = \zerb
\end{equation}
\item
and,
for $\Omega = \Omega(\ab,\bb)$,
\begin{equation}
\Omega ^3 =- (\ab\times\bb)^2 \Omega
\end{equation}
\end{enumerate}

The matrix $\Omega=\Omega(\ub,\vb)$ can be used to simplify the presentation of
both Einstein addition $\ub\op\vb$ and its associated gyration $\gyruvb$,
\begin{equation}
\label{eqEA01}                                 
{\ub}\op{\vb}=\frac{1}{\unpuvc}
\left\{ \ub + \vb -\frac{1}{c^{2}}\frac{\gamma _{{\ub}}}{1+\gamma _{\ub}}
\Omega \ub \right\}
\end{equation}
\begin{equation} \label{gyromega}
\gyruvb = I + \alpha\Omega + \beta\Omega^2
\end{equation}
where $I$ is the $3\times 3$ identity matrix, and where
 \begin{equation} \label{albet}
 \begin{split}
\alpha&=\alpha(\ub,\vb)=-\frac{1}{c^2}
\frac{\gub\gvb(1+\gub+\gvb+\gupvb)}{(1+\gub)(1+\gvb)(1+\gupvb)} \\[6pt]
\beta &=\beta (\ub,\vb)=\quad \frac{1}{c^4}
\frac{\gubs \gvbs}{(1+\gub)(1+\gvb)(1+\gupvb)} \\
 \end{split}
 \end{equation}
satisfying $\alpha < 0$, $\beta > 0$, and
\begin{equation}
\alpha^2+[\ub^2\vb^2-(\ub\ccdot\vb)^2]\beta^2-2\beta=0
\label{002}
\end{equation}
for all $\ub,\vb\in\Rct$.

The gyration matrix $\gyruvb$ in \eqref{gyromega} satisfies the cubic equation
 \begin{equation} \label{trace}
 \begin{split}
\gyr^3 [\ub,\vb]&- \trace(\gyruvb)\gyr^2 [\ub,\vb] \\[3pt]
                &+ \trace(\gyruvb)\gyr [\ub,\vb] - I = 0
 \end{split}
 \end{equation}
called the {\it trace identity}.

The trace identity \eqref{trace}
characterizes $3 \times 3$ matrices that represent
proper rotations of the Euclidean 3-space $\Rt$ about its origin.

The matrix representation of $\gyruvb$ in $\Rt$ relative to an
orthonormal basis is thus an orthogonal $3 \times 3$ matrix with
determinant 1. It follows from \eqref{omegazero} and \eqref{gyromega}
that
\begin{equation}
\gyruvb (\ub \times \vb) = \ub \times \vb
\label{004}
\end{equation}
so that the vector $\ub\times\vb$ lies on the rotation axis of the
gyration $\gyruvb$.

Interesting studies of the trace identity, using analysis, algebra and
geometry is found in an elementary form in
\cite{kalman89} and in a more advanced form in
\cite{gelman01,gelman02,gelman03,gelman04}.

\section{Thomas Precession Graphical Presentation}

 
\begin{figure}[t]  
 \centering         
\psfrag{x}{$x$}
\psfrag{xp}{$x^\prime$}
\psfrag{xpp}{$x^{\prime\prime}$}
\psfrag{y}{$y$}
\psfrag{yp}{$y^\prime$}
\psfrag{ypp}{$y^{\prime\prime}$}
\psfrag{theta}{$\theta$}
\psfrag{epsilon}{$\epsilon$}
\psfrag{Sigma}{$\Sigma$}
\psfrag{Sigmap}{$\Sigma^\prime$}
\psfrag{Sigmapp}{$\Sigma^{\prime\prime}$}
\psfrag{upv}{$\ub\op\vb$}
\psfrag{vpu}{$\vb\op\ub$}
\psfrag{ub}{$\ub$}
\psfrag{vb}{$\vb$}
 \psfrag{---chets1}[]{\hspace{-0.2cm}\lower-1.6ex \hbox {\normalsize{$\blacktriangleright$}}}
 \psfrag{---chets2}[]{\lower-1.2ex \hbox {\normalsize{$\blacktriangleright$}}}
 \psfrag{---chets3}[]{\lower-1.6ex \hbox {\normalsize{$\blacktriangleright$}}}
\includegraphics[width=9cm]{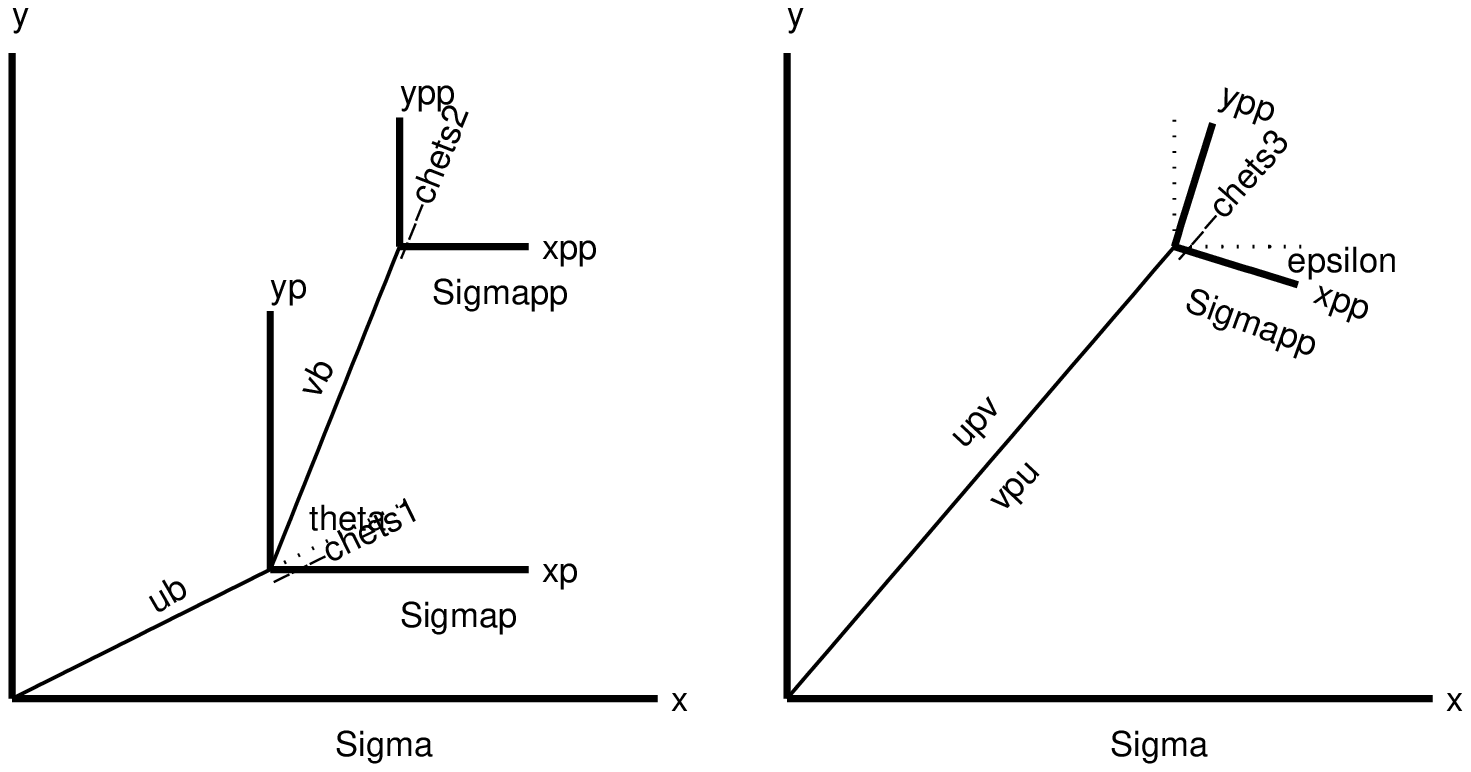}
\caption[The Thomas Precession]{
In the Euclidean plane $\Rtwo$
an inertial frame $\Sigma^{\prime\prime}$ moves uniformly, without rotation,
with velocity $\vb\in\Rctwo$ relative to inertial frame $\Sigma^\prime$.
The latter, in turn, moves uniformly, without rotation,
with velocity $\ub\in\Rctwo$ relative to inertial frame $\Sigma$.
Owing to the presence of Thomas precession,
the  inertial frame $\Sigma^{\prime\prime}$ moves uniformly, with rotation angle $\ep$,
with a composite velocity relative to the inertial frame $\Sigma$.
Is the composite velocity of $\Sigma^{\prime\prime}$ relative to $\Sigma$
$\ub\op\vb$ or $\vb\op\ub$?
The answer is: neither; see \eqref{hvudk}.
The Thomas precession signed angle $\ep$, $-\pi<\ep<\pi$,
turns out to be the unique rotation angle with
rotation axis parallel to $\ub\times\vb$ in $\Rct$
that takes $\vb\op\ub$ into $\ub\op\vb$ according to the
gyrocommutative law $\ub\op\vb=\gyruvb(\vb\op\ub)$.
Being related by \eqref{angle01},
the Thomas precession signed angle $\ep$
and its generating signed angle $\theta$ from $\ub$ and $\vb$
have opposite signs, illustrated
graphically in Figs.\,\ref{fig117}\,--\,\ref{fig118}.
\label{fig115}}
\end{figure}


Let $\Sigma^{\prime\prime}$, $\Sigma^{\prime}$ and $\Sigma$ be
three inertial frames in the Euclidean 3-space $\Rt$
with respective spatial coordinates
$(x^{\prime\prime}, y^{\prime\prime})$, $(x^{\prime}, y^{\prime})$ and $(x,y)$.
The third spatial coordinate of each frame is omitted for simplicity.
Accordingly, these are shown in Fig.~\ref{fig115} in $\Rtwo$ rather than $\Rt$.
Frame $\Sigma^{\prime\prime}$ moves with velocity $\vb\in\Rct$, without rotation,
relative to frame $\Sigma^{\prime}$ which, in turn,
moves with velocity $\ub\in\Rct$, without rotation, relative to frame $\Sigma$.
The angle between $\ub$ and $\vb$ is $\theta$, shown in Fig.~\ref{fig115},
satisfying
\begin{equation} \label{hufd1}
\cos\theta = \frac{\ub}{\|\ub\|} \ccdot \frac{\vb}{\|\vb\|}
\end{equation}
so that, by \eqref{rugh1ds},
\begin{equation} \label{hufd2}
\frac{1}{c^2} \gub\gvb \ub\ccdot\vb = \ggub\ggvb\cos\theta
\end{equation}

Observers at rest relative to $\Sigma$ and
observers at rest relative to $\Sigma^{\prime}$
agree that their coordinates $(x,y)$ and $(x^{\prime}, y^{\prime})$ are parallel.
Similarly,
observers at rest relative to $\Sigma^{\prime}$ and
observers at rest relative to $\Sigma^{\prime\prime}$
agree that their coordinates $(x^{\prime}, y^{\prime})$ and
$(x^{\prime\prime}, y^{\prime\prime})$ are parallel,
as shown in the left part of Fig.~\ref{fig115}.

Counterintuitively, if $\theta\ne0$ and $\theta\ne\pi$,
observers at rest relative to $\Sigma$ and
observers at rest relative to $\Sigma^{\prime\prime}$
agree that their coordinates are not parallel.
Rather, they find that their coordinates are oriented relative to each other
by a Thomas precession angle $\ep$, $0<\ep<\pi$,
as shown in the right part of Fig.~\ref{fig115}.

Let $\ub$ and $\vb$ be two nonzero vectors in the ball $\Rct$.
By the gyrocommutative law in \eqref{laws00},
the gyration $\gyruvb$ takes the composite velocity $\vb\op\ub$ into $\ub\op\vb$.
Indeed, $\gyruvb$ is the unique rotation
with rotation axis parallel to $\ub\times\vb$ that
takes $\vb\op\ub$ into $\ub\op\vb$ through the gyration angle $\ep$,
$0\le\ep<\pi$.
We call $\ep$ the Thomas precession (or, rotation) angle of the gyration $\gyruvb$,
and use the notation
\begin{equation} \label{sukeb}
\ep = \angle\gyruvb
\end{equation}

  
\begin{figure}[t]  
 \centering         
%
\psfrag{pi}[]{$\pi$}
\psfrag{tpi}[]{$2\pi$}
\psfrag{0}[]{$0$}
\psfrag{1}[]{$1$} \psfrag{-1}[]{$-1$}
\psfrag{theta}[]{$\theta$}
\psfrag{costheta}[][][0.5]{$\cos\theta ,~k=1.0$}
\psfrag{cosep}[]{$\cos\epsilon$}
%
 \includegraphics[width=9cm]{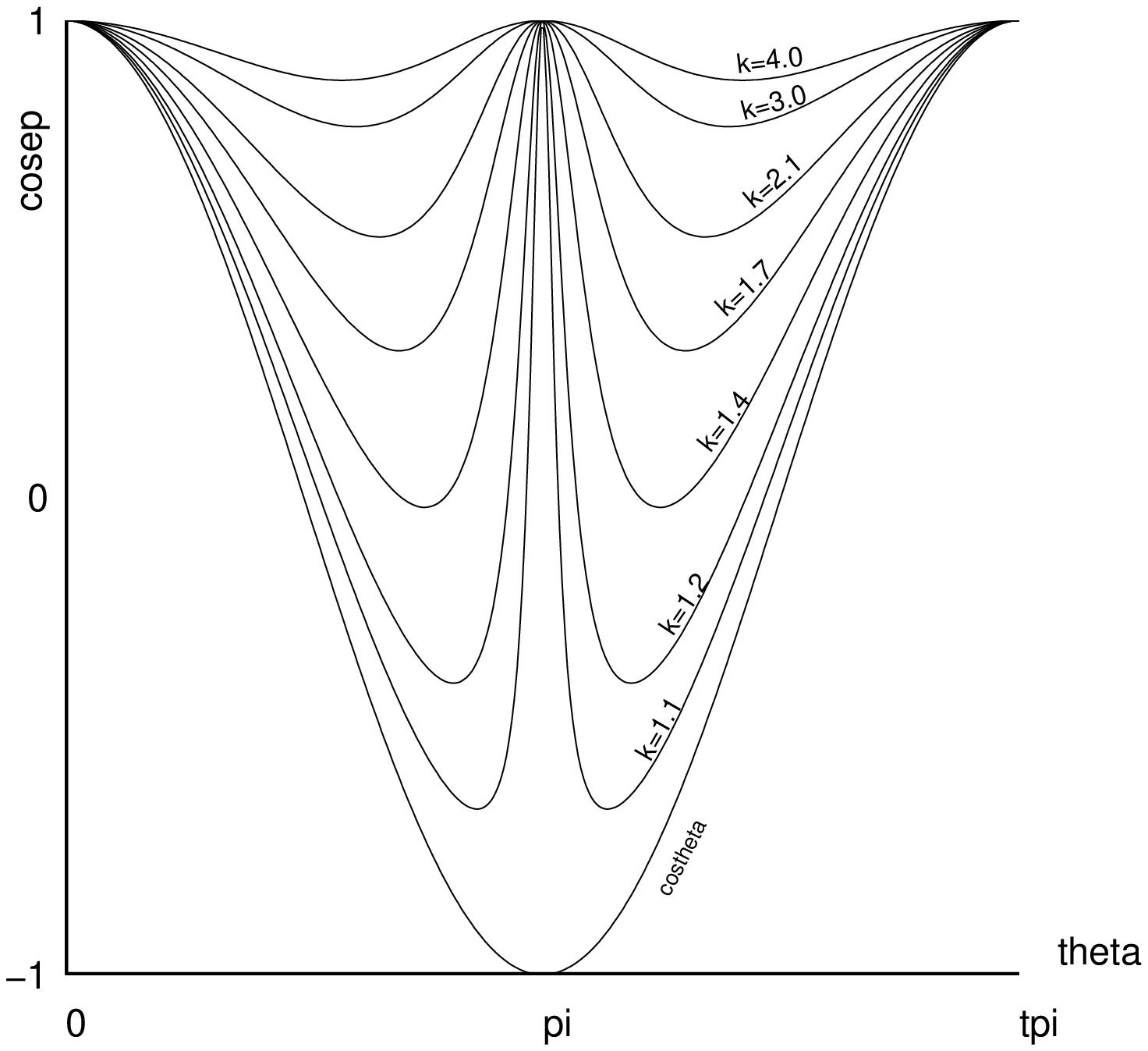}
\caption[Mobius Gyroparallelogram Addition Law]{
A graphical presentation of the cosine of the Thomas precession angle
$\ep$, $\cos\ep$,
\eqref{angle01},
as a function of the angle $\theta$ between its two generating
relativistically admissible velocities $\ub,\vb\in\Rt$ for several
values of $k$, $k$ being a function, \eqref{angle02}, of $\gub$ and $\gvb$.
\label{fig117cm}}
\end{figure}

Accordingly, the Thomas precession angle $\ep=\angle\gyruvb$ generated by
$\ub,\vb\in\Rct$, shown in the right part of Fig.\,\ref{fig115},
satisfies the equations
\begin{equation} \label{angle00}
 \begin{split}
\cos \ep &= \frac{(\ub\op\vb)\cdot(\vb\op\ub)}{\|\ub\op\vb\|^2} \\[8pt]
\sin \ep &= \pm \frac{\|(\ub\op\vb)\times(\vb\op\ub)\|}{\|\ub\op\vb\|^2}
 \end{split}
\end{equation}

The Herculean task of simplifying \eqref{angle00} was accomplished in
\cite{parametrization,noncomm,inner,grouplike}, obtaining
\begin{equation} \label{angle01}
 \begin{split}
\cos \ep &= \frac{(k+\cos\theta)^2-\sin^2\theta}
                  {(k+\cos\theta)^2+\sin^2\theta} \\[3pt]
\sin \ep &= \frac{-2(k+\cos\theta)\,\sin\theta}
                  {(k+\cos\theta)^2+\sin^2\theta} \\
 \end{split}
\end{equation}
where $\theta$, $0\le\theta < 2\pi$, is the angle between the
vectors $\ub,\vb\in\Rt$, forming the horizontal axes in
Figs.\,\ref{fig117cm}--\ref{fig118cm},
and where $k$, $k > 1$, is a velocity parameter given by the equation
\begin{equation} \label{angle02}
k^2= \frac{\gub+1}{\gub-1}\frac{\gvb+1}{\gvb-1}
\end{equation}

  
\begin{figure}[t]  
 \centering         
%
\psfrag{pi}[]{$\pi$}
\psfrag{tpi}[]{$2\pi$}
\psfrag{0}[]{$0$}
\psfrag{1}[]{$1$} \psfrag{-1}[]{$-1$}
\psfrag{theta}[]{$\theta$}
\psfrag{msinep}[]{$-\sin\epsilon$}
%
 \includegraphics[width=9cm]{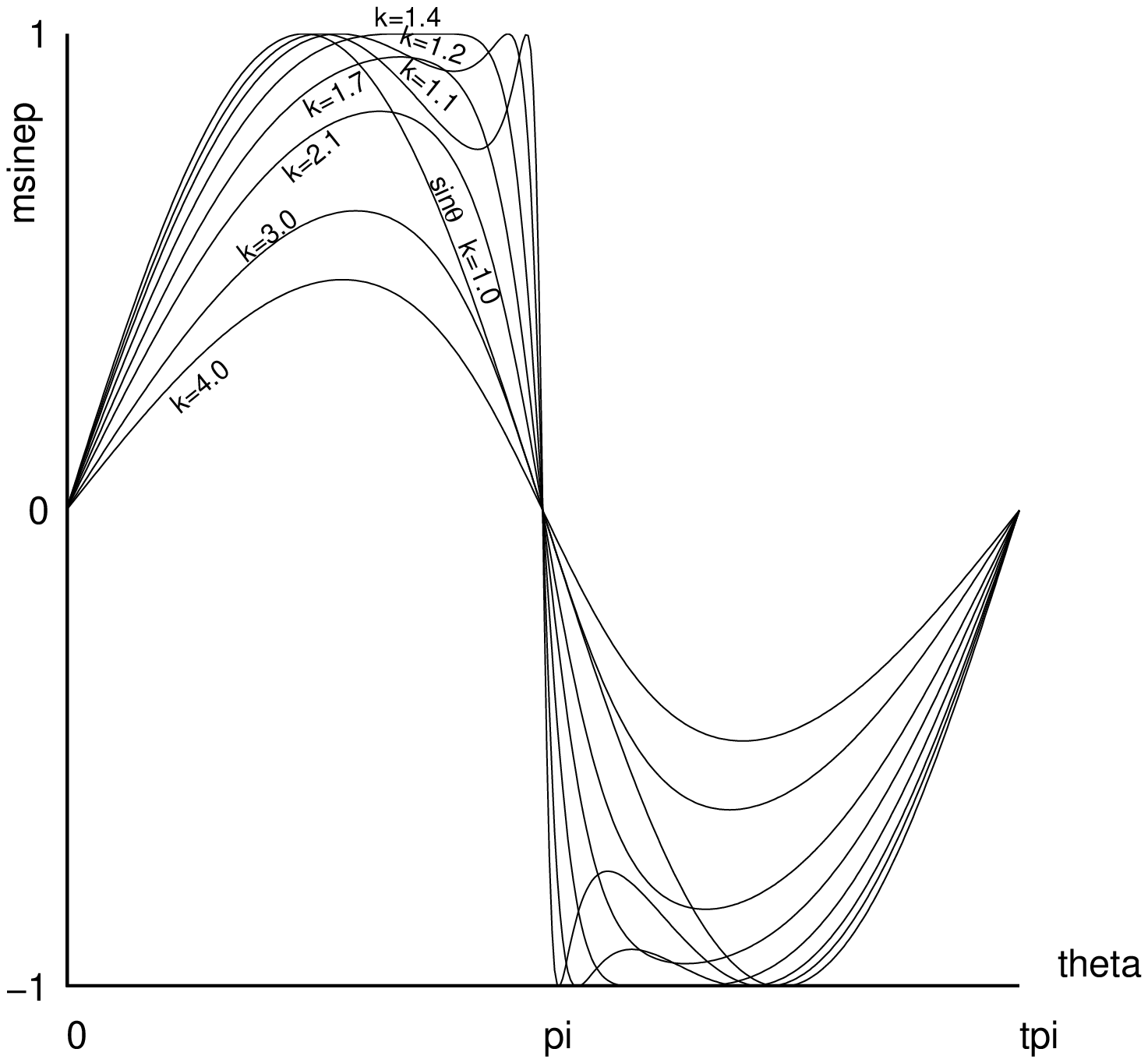}
\caption[Mobius Gyroparallelogram Addition Law]{
A graphical presentation of the negative sine of the Thomas precession angle
$\ep$, $-\sin\ep$,
\eqref{angle01},
as a function of the angle $\theta$ between its two generating
relativistically admissible velocities $\ub,\vb\in\Rt$ for several
values of $k$, $k$ being a function, \eqref{angle02}, of $\gub$ and $\gvb$.
\label{fig118cm}}
\end{figure}
 
The parameter $k$ approaches 1 when both $\|\ub\|$ and $\|\vb\|$
approach $c$. We clearly have the limits
 \begin{equation} \label{limits00}
 \begin{split}
 \lim_{k\rightarrow 1} \cos\ep &= \cos\theta \\
 \lim_{k\rightarrow 1} \sin\ep &=-\sin\theta
 \end{split}
 \end{equation}
for $0 \le \theta \le 2\pi$, $\theta\ne\pi$, seen in
Figs.\,\ref{fig117cm} and \ref{fig118cm}.

Figs.\,\ref{fig117cm} and \ref{fig118cm} present graphically
$\cos\ep$ and $-\sin\ep$ as functions of
$\theta$ for several values of $k$.
As expected, the graphs in these figures
show that for all values of the parameter $k$, $k > 1$, Thomas precession
angle $\ep$ vanishes when $\theta=0$, when $\theta=\pi$, and again,
when $\theta=2\pi$.
In the limit of high relativistic speeds approaching the vacuum speed of
light $c$,
$\|\ub\|,\|\vb\| \rightarrow c$,
the parameter $k$ approaches unity, $k \rightarrow 1$, and
$\ep \rightarrow -\theta$ for all $\theta$ in the punctured interval
$[0,\,\pi)\bigcup(\pi,\,2\pi]$.
The punctured interval is the union of the two connected intervals
$[0,\,\pi)$ and $(\pi,\,2\pi]$ which is the closed connected
interval $[0,2\pi]$ from which the point $\pi$ has been removed.
Thus, there is no Thomas precession angle $\pi$; see also \eqref{eq1g80}.

The extension by abstraction of Thomas precession into gyrations
enables the development of techniques that explain the non-existence
of a gyration whose rotation angle is $\pi$; see the Gyration
Exclusion Theorem in \cite[Theorem 3.36]{mybook03}.

As we see from Figs.\,\ref{fig117cm} and \ref{fig118cm},
the variation of $\ep$ for $0\le\theta\le 2\pi$ is over the interval
$[0,\,2\pi]$ punctured by a $k$-dependent subinterval centered at
$\ep=\pi$.

It is interesting to derive $\cos\frac{\ep}{2}$ and 
$\sin\frac{\ep}{2}$ from \eqref{angle01}:
\begin{equation} \label{angle03}
 \begin{split}
\cos \frac{\ep}{2} &= \pm\sqrt { \frac{1+\cos\ep}{2} }
= \frac{k+\cos\theta}{\sqrt{(k+\cos\theta)^2 +\sin^2\theta}}\\[6pt]
\sin \frac{\ep}{2} &= \pm\sqrt { \frac{1-\cos\ep}{2} }
= -\frac{\sin\theta}{\sqrt{(k+\cos\theta)^2 +\sin^2\theta}}
 \end{split}
\end{equation}
Following Fig.\,\ref{fig115}, the ambiguous signs in
\eqref{angle03} are selected such that $\cos\frac{\ep}{2}>0$ while
$\sin\frac{\ep}{2}$ and $\sin\theta$ have opposite signs.

\section{Thomas Precession Angle}

Thomas precession $\gyruvb$ in \eqref{gyromega}
can be recast into a form familiar
as the representation of a rotation about an
axis by an angle $\ep$,
\begin{equation} \label{eq1d73}
\gyruvb =
\left\{
\begin{array}{ll}
I+\sin\ep\frac{\Omega(\ub,\vb)}{\omega_{\theta}}
+(1-\cos\ep)\frac{\Omega^2(\ub,\vb)}{\omega_{\theta}^2},
& \omega_{\theta}\ne 0 \\
I, & \omega_{\theta}=0
\end{array}
\right.
\end{equation}
where $\ub,\vb\in\Rct$, and where $\ep$ is the Thomas precession angle shown
in Fig.\,\ref{fig115}.

Comparing \eqref{eq1d73} with \eqref{gyromega}, we see that
 \begin{equation} \label{eq1d74}
 \begin{split}
\sin\ep &= \alpha(\ub,\vb)\omega_{\theta} \\[3pt]
1-\cos\ep &= \beta(\ub,\vb)\omega_{\theta} \\
 \end{split}
 \end{equation}
and
\begin{equation} \label{eq1d75}
 \begin{split}
\omega_{\theta}&=\pm\|\ub\times\vb\| \\[3pt]
&= \|\ub\| \|\vb\| \sin\theta \\[3pt]
&= c^2 \frac{\ggub\ggvb}{\gub\gvb} \sin\theta
 \end{split}
\end{equation}
where the ambiguous sign is selected such that $\omega_{\theta}$ and
$\sin\theta$ have equal signs.

It follows from \eqref{eq1d74}\,--\,\eqref{eq1d75}, and from the definition
of $\alpha(\ub,\vb)$ and $\beta(\ub,\vb)$ in \eqref{albet} that
 \begin{equation} \label{eq1d76}
 \begin{split}
\cos\ep &=1- \frac{(\gub-1)(\gvb-1)}{\gupvb+1}\sin^2\theta \\[8pt]
\sin\ep &=
-\frac{\ggub\ggvb+(\gub-1)(\gvb-1)\cos\theta}{\gupvb+1}
\sin\theta
 \end{split}
 \end{equation}

Following \eqref{grbsf09}\,--\,\eqref{rugh2ds} and \eqref{hufd2} we have
\begin{subequations} \label{eq1d77}
\begin{equation} \label{eq1d77a}
\gupvb=\gub\gvb+\ggub\ggvb \,\, \cos\theta
\end{equation}
and
\begin{equation} \label{eq1d77b}
\gmupvb=\gub\gvb-\ggub\ggvb \,\, \cos\theta
\end{equation}
\end{subequations}
so that, by \eqref{eq1d76}\,--\,\eqref{eq1d77a}
\begin{equation} \label{eq1d76s}
 \begin{split}
\cos\ep &=1- \frac{(\gub-1)(\gvb-1)}
{1+\gub\gvb+\ggub\ggvb \cos\theta}
\sin^2\theta \\[8pt]
\sin\ep &=
-\frac{(\gub-1)(\gvb-1)(k+\cos\theta)}
{1+\gub\gvb+\ggub\ggvb \cos\theta}
\sin\theta
 \end{split}
 \end{equation}
where $k>1$ is given by \eqref{angle02}.

The special case when $\ub$ and $\vb$ have equal magnitudes is required for
later reference related to Fig.\,\ref{fig115}.
In this special case $\gub=\gvb$, so that $\ep$ in \eqref{eq1d76s}
reduces to $\ep_s$ given by
\begin{equation} \label{eq1d76st}
\begin{split}
\cos\ep_s &=1- \frac{(\gvb-1)^2 \sin^2\theta}
{1+\gvbs+(\gvbs-1)\cos\theta}
\\[8pt]
\sin\ep_s &=
-\frac{(\gvbs-1)+(\gvb-1)^2\cos\theta}
{1+\gvbs+(\gvbs-1)\cos\theta}
\sin\theta
 \end{split}
 \end{equation}

Solving \eqref{eq1d77} for $\cos\theta$ we obtain the equations
 \begin{equation} \label{eq1d78}
 \begin{split}
\cos\theta&=\frac{\gupvb-\gub\gvb}{\ggub\ggvb}
= \frac{-\gamma_{\om\ub\op\vb}^{\phantom{1}} + \gub\gvb}
{\ggub\ggvb}
\\[6pt]
\sin^2 \theta &=1-\cos^2\theta \\[6pt]
&=\frac{1-\gubs-\gvbs-\gupvbs+2\gub\gvb\gupvb}{(\gubs-1)(\gvbs-1)}
 \end{split}
 \end{equation}
The substitution of \eqref{eq1d78} into \eqref{eq1d76} gives
 \begin{equation} \label{eq1d79}
 \begin{split}
\cos\ep&= \frac{1} {(\gub+1)(\gvb+1)(\gupvb+1)} \\[3pt]
&\times \{-\gub\gvb\gupvb+\gubs+\gvbs+\gupvbs \\[3pt]
&+\gub\gvb+\gub\gupvb+\gvb\gupvb+\gub+\gvb+\gupvb\}
 \end{split}
 \end{equation}
so that, finally, we obtain the elegant expression
\begin{equation} \label{eq1g80}
1+\cos\ep = \frac{(1+\gub+\gvb+\gupvb)^2}{(1+\gub)(1+\gvb)(1+\gupvb)} > 0
\end{equation}
which agrees with McFarlane's result,
cited in \cite[Eq.~(2.10.7]{urbantkebookeng}.
It implies that $\ep\ne\pi$ for all $\ub,\vb\in\Rct$; and that
\begin{equation} \label{eq1d80}
\cos\frac{\ep}{2}=\sqrt{\frac{1+\cos\ep}{2}}
=\frac{1+\gub+\gvb+\gupvb}{\sqrt{2}\sqrt{1+\gub}\sqrt{1+\gvb}\sqrt{1+\gupvb}}
\end{equation}

Finally, we also have the elegant identity
\begin{equation} \label{dugit4}
\tan^2 \frac{\epsilon}{2} = \left(
\frac{\sin\epsilon}{1+\cos\epsilon} \right)^2 =
\frac{
1+2\gub\gvb\gupvb -\gubs-\gvbs-\gupvbs
}{
(1+\gub+\gvb+\gupvb)^2
}
\end{equation}

Hence, \eqref{sukeb},
the Thomas precession angle $\ep = \angle\gyruvb$ in Fig.\,\ref{fig115}
is given by the equation
\begin{equation} \label{dugit5}
\tan^2  \frac{\angle\gyruvb}{2} =
\frac{
1+2\gub\gvb\gupvb -\gubs-\gvbs-\gupvbs
}{
(1+\gub+\gvb+\gupvb)^2
}
\end{equation}
Noting the gyration even property in \eqref{laws00} and replacing
$\ub$ by $\om\ub$ in \eqref{dugit5}, we obtain the equations
\begin{equation} \label{dugit6}
\begin{split}
\tan^2  \frac{\angle\gyr[\ub,\om\vb]}{2} &=
\tan^2  \frac{\angle\gyr[\om\ub,\vb]}{2}
\\[8pt] &=
\frac{
1+2\gub\gvb\gmupvb -\gubs-\gvbs-\gmupvbs
}{
(1+\gub+\gvb+\gmupvb)^2
}
\\[8pt] &= \tan^2\tfrac{\delta}{2}
\end{split}
\end{equation}
The extreme right-hand side of \eqref{dugit6} follows from \eqref{hdke3}\,--\,\eqref{hdke4}.

Interestingly, it follows from \eqref{dugit6}
that the Thomas precession angle
generated by $\ub$ and $\om\vb$, that is,
$\angle\gyr[\ub,\om\vb]$,
possesses
an important hyperbolic geometric property. It equals the defect $\delta$
of the gyrotriangle generated by $\ub$ and $\vb$ in $\Rct$;
see also \cite[pp.~236--237]{mybook01} and the Gyration -- Defect Theorem
in \cite[Theorem 8.55, p.~317]{mybook03}.

The gyration $\gyr[\ub,\om\vb]$ possesses an important
gyroalgebraic property as well.
It gives rise to a second binary operation $\sqp$, called {\it Einstein coaddition},
given by the equation
\begin{equation} \label{hvudk}
\ub\sqp\vb = \ub\op\gyr[\ub,\om\vb]\vb
\end{equation}
which can be dualized into the equation \cite[Theorem 2.14]{mybook03}
\begin{equation} \label{eqdual13}
\ub\op\vb = \ub\sqp \gyruvb\vb
\end{equation}

Unlike Einstein addition, which is gyrocommutative,
Einstein coaddition is commutative.
Furthermore, it possesses a geometric interpretation as a
{\it gyroparallelogram addition law}, and it gives rise to the two
mutually dual right cancellation laws \cite{mybook03}
\begin{equation} \label{rightcanc}
\begin{split}
(\vb\op \ub) \sqm \ub &= \vb \\
(\vb\sqp \ub) \om \ub &= \vb
\end{split}
\end{equation}

Einstein coaddition $\sqp$ captures useful analogies with classical results,
two of which are the right cancellation laws in \eqref{rightcanc}.
Another important analogy that Einstein coaddition captures is associated
with the gyromidpoint. Indeed,  the midpoint $M_{_{A,B}}$ shown in Fig.~\ref{fig163aceuc3m}
is expressed in terms of the points $A$ and $B$ by the equation
\begin{equation} \label{guliver1}
M_{_{A,B}} = A+(-A+B)\half = \half(A+B)
\end{equation}
while, in full analogy, the gyromidpoint $M_{_{A,B}}$ shown in Fig.~\ref{fig163acein3m}
is expressed in terms of the points $A$ and $B$ by the equation \cite{mybook03}
\begin{equation} \label{guliver2}
M_{_{A,B}} = A\op(\om A\op B)\od\half = \half\od(A\sqp B)
\end{equation}

The right part of Fig.~\ref{fig115} raises the question as to whether the
composite velocity of frame $\Sigma^{\prime\prime}$ relative to frame $\Sigma$
is $\ub\op\vb$ or $\vb\op\ub$.
The answer is that the
composite velocity of frame $\Sigma^{\prime\prime}$ relative to frame $\Sigma$
is neither $\ub\op\vb$ nor $\vb\op\ub$.
Rather, it is given by the commutative composite velocity $\ub\sqp\vb$.
Indeed, it is demonstrated in \cite[Chap.~10--Epilogue]{mybook05},
and in more details in \cite[Chap.~13]{mybook03}, that
looking at the relativistic velocity addition law and its underlying
hyperbolic geometry through the lens of the cosmological stellar aberration effect
leads to a startling conclusion: relativistic velocities are gyrovectors
that add in the cosmos
according to the gyroparallelogram addition law of hyperbolic geometry,
that is, according to the commutative addition $\ub\sqp\vb$,
rather than either Einstein addition $\ub\op\vb$ or $\vb\op\ub$.

\section{Thomas Precession Frequency}

Let us consider a spinning spherical object moving with velocity $\vb$
of uniform magnitude $v=\|\vb\|$
along a circular path in some inertial frame $\Sigma$. We assume that the spin axis
lies in the plane containing the circular orbit, as shown in Fig.~\ref{fig294m}.
The spinning object acts
like a gyroscope, maintaining the direction of its spin axis
in the transition from one inertial frame into another one, as seen
by inertial observers moving instantaneously with the
accelerated object. Following
Taylor and Wheeler, we approximate the
circular path by a regular polygon of $n$ sides \cite{taylorwheeler}, as shown
in Fig.\,\ref{fig294m} for $n=8$.
In moving once
around this orbit the object moves with uniform
velocity $\vb$
in straight-line paths interrupted by $n$ sudden changes of direction, each
through an angle $\theta_n=2\pi/n$.

 
\begin{figure}[t]  
 \centering         
 \psfrag{A}{$A,~\Sigma$}
 \psfrag{B}{$B$}
 \psfrag{C}{$C$}
 \psfrag{e}{$\theta_n$}
 \psfrag{S1}{$\Sigma^\prime$}
 \psfrag{S2}{$\Sigma^{\prime\prime}$}
 \psfrag{Initial}{$\rm{Initial~Spin}$}
 \psfrag{Final}{$\rm{Final~Spin}$}
 \includegraphics[width=9cm]{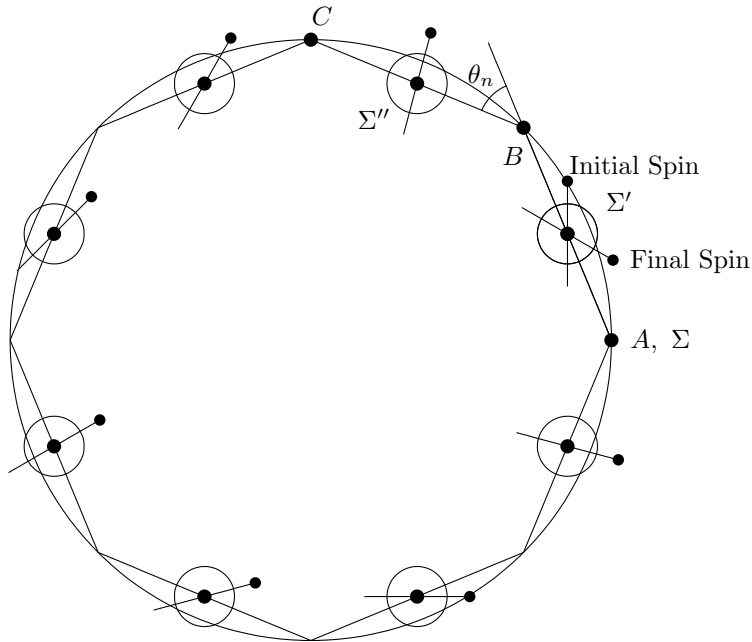}
\caption[Mobius Gyroparallelogram Addition Law]{
A regular polygonal path in $\Rt$ as an approximation to the Newtonian
circular path of a spinning spherical object.
The change of direction at each vertex of the polygon is $\theta_n=2\pi/n$,
where $n$ is the number of the polygon sides. Here, $n=8$.
In the limit $n\rightarrow\infty$, the polygonal path tends to the circular path.
A spinning spherical object is moving with velocity of uniform magnitude along
the polygonal path. The points $A,B,C\in\Rt$ are three adjacent vertices of the polygon
in the rest (laboratory) frame $\Sigma$. When the object moves from $A$ to $B$
it is at rest relative to the frame $\Sigma^\prime$, and
when the object moves from $B$ to $C$
it is at rest relative to the frame $\Sigma^{\prime\prime}$.
The relationship between the three inertial frames
$\Sigma$, $\Sigma^\prime$ and $\Sigma^{\prime\prime}$ is thus the one shown in
Fig.~\ref{fig115} with $\theta=\theta_n$.
Accordingly, since the object moves in the counterclockwise direction,
it precesses in the clockwise direction.
\newline
Initially, the spin of the object is vertical when the object moves uniformly from
$A$ to $B$. After completing its first closed orbit in the counterclockwise direction,
the object returns to is original position,
now moving from $A$ to $B$ with a spin that is precessed in the
clockwise direction.
The initial spin and the final spin for the first closed orbit starting at $A$
are shown.
\label{fig294m}}
\end{figure}

An observer at rest relative to the laboratory frame
$\Sigma$ views the motion of the object along the
polygonal path as the result of successive
boosts (A boost being a Lorentz transformation without rotation \cite{quasi91});
see Sec.~\ref{davidi}.
He therefore measures a Thomas precession angle $\ep_n$ by which the object
spin axis is precessed when the object rounds a corner. By \eqref{eq1d76st}, this
Thomas precession angle $\ep_n$ is determined by the equations
\begin{equation} \label{eq1d76m}
\begin{split}
\cos\ep_n &=1- \frac{(\gvb-1)^2 \sin^2\frac{2\pi}{n}}
{(\gvbs+1)+(\gvbs-1)\cos\frac{2\pi}{n}}
\\[6pt]
\sin\ep_n &=
-\frac{(\gvbs-1)+(\gvb-1)^2\cos\frac{2\pi}{n}}
{(\gvbs+1)+(\gvbs-1)\cos\frac{2\pi}{n}}
\sin\frac{2\pi}{n}
 \end{split}
 \end{equation}

By Euler's equation we have
\begin{equation} \label{kurni94}
\begin{split}
e^{i\ep_n} &= \cos\ep_n + i \sin\ep_n
\\[6pt] &= 1 + f(\frac{2\pi}{n})
\end{split}
\end{equation}
where $i=\sqrt{-1}$ and
\begin{equation} \label{kurni95}
f(\phi) = - \frac{
(\gvb-1)^2\sin \phi + i\{(\gvbs-1)+(\gvb-1)^2\cos \phi \}
}{
2+(\gvbs-1)(1+\cos \phi)
}
\sin \phi
\end{equation}
$\phi\in\Rb$.

As the spinning object moves around its polygonal orbit, its spin axis, as
observed in $\Sigma$, precesses by the Thomas precession angle $\ep_n$ when it rounds each
of the $n$ corners of the polygon as shown in Fig.~\ref{fig294m}.
The total angle of
precession is thus $n\ep_n$, represented by the unimodular complex
number
\begin{equation} \label{kurni96}
e^{in\ep_n} = \left\{ 1 + f(\frac{2\pi}{n}) \right\}^n
\end{equation}

In the limit $n\rightarrow\infty$ the polygonal path becomes a circular path,
and the frame of reference in which the center of momentum of the spinning object
is momentarily at rest is being changed continually. The
total Thomas precession is thus
the angle $\ep_t$ given by the equation
\begin{equation} \label{kurni97}
e^{i\ep_t} = \lim_{n\rightarrow \infty} e^{in\ep_n}
= \lim_{n\rightarrow \infty}
\left\{ 1 + f(\frac{2\pi}{n}) \right\}^n
\end{equation}

Let $g(x)$, $x\in\Rb$, be the function
\begin{equation} \label{kurni98}
g(x) = \lim_{n\rightarrow \infty} \left\{ 1 + f(\frac{x}{n}) \right\}^n
\end{equation}

The function $f(\phi)$ is continuous on $\Rb$, satisfying $f(0)=0$.
Hence,
\begin{equation} \label{kurni98s}
\lim_{n\rightarrow \infty} f(\frac{x}{n}) = 0
\end{equation}
for any $x\in\Rb$.

Interchanging the limit in \eqref{kurni98} with a differentiation with respect to $x$
we find that the function $g(x)$ satisfies the initial value problem
\begin{equation} \label{kurni99}
\begin{split}
g^\prime(x) &= f^\prime(0) g(x) \\
g(0) &= 1
\end{split}
\end{equation}
for $x\in\Rb$.

The unique solution of the initial value problem \eqref{kurni99} is
\begin{equation} \label{kurni100}
g(x) = e^{f^\prime(0)x}
\end{equation}
Hence, in particular for $x=2\pi$, it follows from
\eqref{kurni97}, \eqref{kurni98} and \eqref{kurni100} that
\begin{equation} \label{kurni101}
e^{i\ep_t} = g(2\pi) = e^{2\pi f^\prime(0)}
\end{equation}

But,
\begin{equation} \label{kurni101s}
f^\prime(0) = -i \frac{\gvb-1}{\gvb}
\end{equation}
Hence, by \eqref{kurni101}\,--\,\eqref{kurni101s},
the Thomas precession angle $\ep_t$ is given by the equation
\begin{equation} \label{kurni102}
\ep_t = -2\pi \frac{\gvb-1}{\gvb}
\end{equation}

The Thomas precession angle $\ep_t$ is the angle through which the spin axis
precesses in one complete circular orbit. It requires, therefore,
$2\pi / \ep_t$ orbits for the object to precess to its original
orientation through $2\pi$ radians. Hence, if the
angular velocity of the circular motion of the object
is $\omega$, then the angular velocity $\omega_t$ of the
Thomas precession angle of the object is given by the equation
\begin{equation} \label{kurni103}
\omega_t = \frac{\ep_t}{2\pi} \omega = - \frac{\gvb-1}{\gvb} \omega
\end{equation}

The quantity $\omega_t$ in \eqref{kurni103} is the angular velocity of the
Thomas precession angle $\ep_t$ of a particle that moves in a circular orbit with
angular velocity $\omega$.
Eq.~\eqref{kurni103} relates the angular velocity $\omega_t$ of the Thomas
precession angle $\ep_t$ to its generating angular velocity $\omega$.
It demonstrates that the angular velocities $\omega_t$ and $\omega$
are oppositely directed, as shown graphically in Fig.~\ref{fig115}.

If the magnitude of the velocity $\vb$ and the acceleration $\ab$ of the spinning
object are $v$ and $a$ then its angular velocity is given by the equation
$\omega = a/v$. Hence,
the angular velocity $\omega_t$ of the Thomas precession angle $\ep_t$ is
given by the equation
\begin{equation} \label{kurni107}
\omega_t = - \frac{\gvb-1}{\gvb} \frac{a}{v}
\end{equation}

Taking into account the direction of the Thomas precession axis and the
velocity and the acceleration of the spinning object,
and noting \eqref{rugh1ds}, the Thomas precession angular velocity $\omega_t$
in \eqref{kurni107} can be written as a vector equation,
\begin{equation} \label{kurni108}
\omegab_t = \frac{\gvb-1}{\gvb} \frac{\ab\times\vb}{v^2}
=  \frac{\gvb}{1+\gvb} \frac{\ab\times\vb}{c^2}
\end{equation}

The coordinate axes in the rest frame of any body in torque-free,
accelerated motion precesses with respect to the laboratory axes with
an angular velocity $\omegab_t$ is given by \eqref{kurni108}.
Since
$\gvb/(1+\gvb)=1/2+(1/8)(v^2/c^2)+\ldots$,
the angular velocity $\omegab_t$ of the resulting Thomas precession,
for the case when $v=\|\vb\| <\!\!< c$,
is given approximately by the equation
\begin{equation} \label{kurni109}
\omegab_t = \frac{1}{2} \frac{\ab\times\vb}{c^2}
\end{equation}

According to Herbert Goldstein \cite[p.~288]{goldstein}, $\omegab_t$ in \eqref{kurni109}
is known as the {\it Thomas precession frequency}.

The Thomas precession frequency \eqref{kurni109} involves the famous factor $1/2$,
known as {\it Thomas half}.
The experimental significance of this factor is well known:
The spinning electron of the Goudsmit-Uhlenbeck model
gives twice the observed precession effect, which is reduced to the
observed one by means of the Thomas half \cite{thomas82}.

\section{Thomas Precession and Boost Composition}\label{davidi}

Einstein addition underlies the Lorentz transformation group of
special relativity.
A Lorentz transformation is a linear transformation of spacetime coordinates that
fixes the spacetime origin.
A Lorentz boost, $B(\vb)$, is a Lorentz transformation without rotation,
parametrized by a velocity parameter $\vb\in \Rct$.
The velocity parameter is given by its components, $\vb=(v_1,v_2,v_3)$,
with respect to a given Cartesian coordinate system of $\Rct$.
Being linear, the Lorentz boost has a matrix representation,
which turns out to be \cite{moller52},
\begin{equation}\label{lormatrix}
\begin{split}
&B(\vb) = \\[4pt]
&\begin{pmatrix}
\gvb & \cmt\gvb   v_1 & \cmt\gvb   v_2 & \cmt\gvb   v_3 \\[6pt]
\gvb   v_1 & 1+\cmt\frac{\gvbs}{\gvb+1}   v_1^2 &
\cmt\frac{\gvbs}{\gvb+1}   v_1   v_2            &
\cmt\frac{\gvbs}{\gvb+1}   v_1   v_3                    \\[6pt]
\gvb   v_2 &\cmt\frac{\gvbs}{\gvb+1}   v_1   v_2&
1+\cmt\frac{\gvbs}{\gvb+1}   v_2^2 &
\cmt\frac{\gvbs}{\gvb+1}   v_2   v_3                    \\[6pt]
\gvb   v_3 &\cmt\frac{\gvbs}{\gvb+1}   v_1   v_3&
\cmt\frac{\gvbs}{\gvb+1}   v_2   v_3            &
1+\cmt\frac{\gvbs}{\gvb+1}   v_3^2
\end{pmatrix}
\end{split}
\end{equation}

Employing the matrix representation \eqref{lormatrix} of the
Lorentz transformation boost,
the Lorentz boost application to spacetime coordinates
takes the form
\begin{equation} \label{eqaa02s3}
B(\vb) \begin{pmatrix} t \\ \xb \end{pmatrix}
=
B(\vb)
\begin{pmatrix} t \\[3pt] x_1 \\[3pt] x_2  \\[3pt] x_3  \end{pmatrix}
=:
\begin{pmatrix}
t^\prime \\[3pt] x_1^\prime \\[3pt] x_2^\prime  \\[3pt] x_3^\prime
\end{pmatrix}
=
\begin{pmatrix}
t^\prime \\[3pt] \xb^\prime
\end{pmatrix}
\end{equation}
where $\vb=(v_1,v_2,v_3)^t\in \Rct$,
$\xb=(x_1,x_2,x_3)^t\in \Rt$,
$\xb^\prime=(x_1^\prime,x_2^\prime,x_3^\prime)^t\in \Rt$,
and $t,t^\prime\in \Rb$, where exponent $t$ denotes transposition.

A 2-dimensional boost is obtained in the special case when $v_3=x_3=0$
in \eqref{lormatrix}\,--\,\eqref{eqaa02s3}.
For simplicity, the boosts of inertial frames in Fig.~\ref{fig115} are two dimensional
boosts, and time coordinates are not shown.
In this figure, the spacetime coordinate systems
$\Sigma$, $\Sigma^\prime$ and $\Sigma^{\prime\prime}$ (only two space coordinates
are shown) are related by boosts.
Specifically, in Fig.~\ref{fig115},
\begin{enumerate}
\item\label{kamoti1}
the application of the boost $B(\ub)$ to the
spacetime coordinate system $\Sigma$ gives the spacetime coordinate system $\Sigma^\prime$,
\item\label{kamoti2}
the application of the boost $B(\vb)$ to the
spacetime coordinate system $\Sigma^\prime$
gives the spacetime coordinate system $\Sigma^{\prime\prime}$,
\item\label{kamoti3}
and the application of the boost $B(\ub\op\vb)$, or $B(\vb\op\ub)$,
to the
spacetime coordinate system $\Sigma$, preceded, or followed respectively,
by a Thomas precession
(see \eqref{guliv02}\,--\,\eqref{guliv03} in Theorem \ref{mainthm1} below)
gives the spacetime coordinate system $\Sigma^{\prime\prime}$.
\end{enumerate}

The Lorentz boost \eqref{lormatrix}\,--\,\eqref{eqaa02s3}
can be written vectorially in the form
\begin{equation} \label{kyhd02}
B(\ub)\tonxb =
\begin{pmatrix}
\gub(t+\frac{1}{c^2}\ub\ccdot \xb) \\[6pt]
\gub\ub t + \xb + \frac{1}{c^2}\frac{\gubs}{1+\gub}
(\ub\ccdot \xb)\ub
\end{pmatrix}
\end{equation}

Rewriting \eqref{kyhd02} with $\xb=\vb t\in\Rt$, $\ub,\vb\in \Rct\subset\Rt$,
we have
\begin{equation} \label{kyhd03}
\begin{split}
B(\ub)  \begin{pmatrix} t \\ \xb   \end{pmatrix}
=
B(\ub)  \begin{pmatrix} t \\ \vb t \end{pmatrix}
&=
\begin{pmatrix}
\gub(t+\frac{1}{c^2}\ub\ccdot \vb t) \\[6pt]
\gub\ub t + \vb t + \frac{1}{c^2}\frac{\gubs}{1+\gub}
(\ub\ccdot \vb t)\ub
\end{pmatrix}
\\[8pt] &=
\begin{pmatrix} \frac{\gamma_{\ub\op \vb}^{\phantom{O}}}{\gvb} t\\[6pt]
\frac{\gamma_{\ub\op \vb}^{\phantom{O}}}{\gvb}
(\ub\op \vb)t
 \end{pmatrix}
=
\begin{pmatrix} t^\prime \\ \xb^\prime   \end{pmatrix}
\end{split}
\end{equation}

The equations in \eqref{kyhd03} reveal explicitly the way Einstein velocity addition underlies
the Lorentz boost.
The third equality in \eqref{kyhd03} follows from
\eqref{grbsf09p0} and \eqref{eq01}.

In general, the composition of two boosts is equivalent to a single boost
preceded, or followed, by the space rotation that Thomas precession generates,
as we see from the following theorem:

\begin{ttheorem}\label{mainthm1}
{\bf (The Boost Composition Theorem):}
Let $\ub,\vb,\wb\in\Rct$ be relativistically admissible velocities,
let $\xb=\wb t$, $t>0$, and
let $B(\ub)$ and $B(\vb)$ be two boosts. Furthermore,
let $\Gyr[\ub,\vb]$ be the spacetime gyration of space coordinates,
given by
\begin{equation} \label{guliv01}
\Gyr[\ub,\vb] \begin{pmatrix} t \\ \xb=\wb t   \end{pmatrix}
:= \begin{pmatrix} t \\ (\gyruvb\wb)t   \end{pmatrix}
\end{equation}

Then, boost composition is given by each of the two equations
\begin{equation} \label{guliv02}
B(\ub)B(\vb) = B(\ub\op\vb)\Gyr[\ub,\vb]
\end{equation}
\begin{equation} \label{guliv03}
B(\ub)B(\vb) = \Gyr[\vb,\ub]B(\vb\op\ub)
\end{equation}
\end{ttheorem}
\begin{proof}
We will show that \eqref{guliv02} follows from the gyroassociative law
of Einstein addition and that that \eqref{guliv03} follows from \eqref{guliv02}
and the gyrocommutative law of Einstein addition.

On the one hand we have
the chain of equations below, which are numbered for subsequent derivation:
\begin{equation} \label{guliv04}
\begin{split}
B(\ub)B(\vb) \begin{pmatrix} t \\ \xb   \end{pmatrix}
&
\overbrace{=\!\!=\!\!=}^{(1)} \hspace{0.2cm}
B(\ub)B(\vb) \begin{pmatrix} t \\ \wb t   \end{pmatrix}
\\[8pt]&
\overbrace{=\!\!=\!\!=}^{(2)} \hspace{0.2cm}
B(\ub) \begin{pmatrix}
\frac{\gamma_{\vb\op\wb}^{\phantom{1}}}{\gamma_{\wb}^{\phantom{1}}}
t
\\
\frac{\gamma_{\vb\op\wb}^{\phantom{1}}}{\gamma_{\wb}^{\phantom{1}}}
(\vb\op\wb)t
\end{pmatrix}
\\[8pt]&
\overbrace{=\!\!=\!\!=}^{(3)} \hspace{0.2cm}
B(\ub) \begin{pmatrix} t^\prime \\ (\vb\op\wb)  t^\prime \end{pmatrix}
\\[8pt]&
\overbrace{=\!\!=\!\!=}^{(4)} \hspace{0.2cm}
\begin{pmatrix}
\frac{\gamma_{\ub\op(\vb\op\wb)}^{\phantom{1}}}{\gamma_{\vb\op\wb}^{\phantom{1}}}
t^\prime \\
\frac{\gamma_{\ub\op(\vb\op\wb)}^{\phantom{1}}}{\gamma_{\vb\op\wb}^{\phantom{1}}}
\{ \ub\op(\vb\op\wb) \} t^\prime
\end{pmatrix}
\\[8pt]&
\overbrace{=\!\!=\!\!=}^{(5)} \hspace{0.2cm}
\begin{pmatrix}
\frac{\gamma_{\ub\op(\vb\op\wb)}^{\phantom{1}}}{\gamma_{\wb}^{\phantom{1}}}
t \\
\frac{\gamma_{\ub\op(\vb\op\wb)}^{\phantom{1}}}{\gamma_{\wb}^{\phantom{1}}}
\{ \ub\op(\vb\op\wb) \} t
\end{pmatrix}
\\[8pt]&
\overbrace{=\!\!=\!\!=}^{(6)} \hspace{0.2cm}
\begin{pmatrix}
\frac{\gamma_{(\ub\op\vb)\op\gyruvb\wb}^{\phantom{1}}}{\gamma_{\wb}^{\phantom{1}}}
t \\
\frac{\gamma_{(\ub\op\vb)\op\gyruvb\wb}^{\phantom{1}}}{\gamma_{\wb}^{\phantom{1}}}
\{(\ub\op\vb)\op\gyruvb\wb \} t
\end{pmatrix}
\end{split}
\end{equation}
Derivation of the numbered equalities in \eqref{guliv04} follows:
\begin{enumerate}
\item\label{kamotu1}
Follows from the definition $\xb=\wb t$.
\item\label{kamotu2}
Follows from \ref{kamotu1} by a boost application to spacetime coordinates
according to \eqref{kyhd03}.
\item\label{kamotu3}
Follows from \ref{kamotu2} by the obvious definition
\begin{equation} \label{guliv05}
t^\prime = \frac{\gamma_{\vb\op\wb}^{\phantom{1}}}{\gamma_{\wb}^{\phantom{1}}} t
\end{equation}
\item\label{kamotu4}
Follows from \ref{kamotu3} by a boost application to spacetime coordinates
according to \eqref{kyhd03}.
\item\label{kamotu5}
Follows from \ref{kamotu4} by the substitution of \eqref{guliv05} for
$t^\prime$.
\item\label{kamotu6}
Follows from \ref{kamotu5} by the gyroassociative law of Einstein addition.
\end{enumerate}

On the other hand we have
the chain of equations below, which are numbered for subsequent derivation:
\begin{equation} \label{guliv06}
\begin{split}
B(\ub\op\vb)\Gyr[\ub,\vb]  \begin{pmatrix} t \\ \xb   \end{pmatrix}
&
\overbrace{=\!\!=\!\!=}^{(1)} \hspace{0.2cm}
B(\ub\op\vb) \begin{pmatrix} t \\ \gyruvb\xb   \end{pmatrix}
\\[8pt]&
\overbrace{=\!\!=\!\!=}^{(2)} \hspace{0.2cm}
B(\ub\op\vb) \begin{pmatrix} t \\ (\gyruvb\wb)t   \end{pmatrix}
\\[8pt]&
\overbrace{=\!\!=\!\!=}^{(3)} \hspace{0.2cm}
\begin{pmatrix}
\frac{
\gamma_{(\ub\op\vb)\op\gyruvb\wb}^{\phantom{1}}
}{
\gamma_{\gyruvb\wb}^{\phantom{1}}
} t
\\
\frac{
\gamma_{(\ub\op\vb)\op\gyruvb\wb}^{\phantom{1}}
}{
\gamma_{\gyruvb\wb}^{\phantom{1}}
}
\{ (\ub\op\vb) \op \gyruvb\wb\} t
\end{pmatrix}
\\[8pt]&
\overbrace{=\!\!=\!\!=}^{(4)} \hspace{0.2cm}
\begin{pmatrix}
\frac{
\gamma_{(\ub\op\vb)\op\gyruvb\wb}^{\phantom{1}}
}{
\gamma_{\wb}^{\phantom{1}}
} t
\\
\frac{
\gamma_{(\ub\op\vb)\op\gyruvb\wb}^{\phantom{1}}
}{
\gamma_{\wb}^{\phantom{1}}
} 
\{ (\ub\op\vb) \op \gyruvb\wb\} t
\end{pmatrix}
\end{split}
\end{equation}
Derivation of the numbered equalities in \eqref{guliv06} follows:
\begin{enumerate}
\item\label{kamota1}
Follows from the definition of the spacetime gyration $\Gyr[\ub,\vb]$ in terms
of the space gyration $\gyruvb$ in \eqref{guliv01}.
\item\label{kamota2}
Follows from \ref{kamota1} by definition, $\xb=\wb t$.
\item\label{kamota3}
Follows from \ref{kamota2} by a boost application to spacetime coordinates
according to \eqref{kyhd03}.
\item\label{kamota4}
Follows from \ref{kamota3} by the identity
$\gwb=\gamma_{\gyruvb\wb}^{\phantom{1}}$ that, in turn, follows from
the definition of gamma factors in \eqref{v72gs}
along with the invariance \eqref{eq005a}
of relativistically admissible velocities under gyrations.
\end{enumerate}

The extreme right-hand sides of the chain of equations
\eqref{guliv04} and \eqref{guliv06} are identically equal.
Hence, the extreme left-hand sides of \eqref{guliv04} and \eqref{guliv06} are
equal for all
spacetime events $(t,\xb)^t$, $t\in\Rb$, $\xb=\wb t$, $\wb\in\Rct$,
thus verifying \eqref{guliv02}.

In order to verify \eqref{guliv03}, let us now consider
the chain of equations below, which are numbered for subsequent derivation:
\begin{equation} \label{guliv07}
\begin{split}
\Gyr[\vb,\ub] B(\vb\op\ub) \begin{pmatrix} t \\ \xb   \end{pmatrix}
&
\overbrace{=\!\!=\!\!=}^{(1)} \hspace{0.2cm}
\Gyr[\vb,\ub] B(\vb\op\ub) \begin{pmatrix} t \\ \wb t   \end{pmatrix}
\\[8pt]&
\overbrace{=\!\!=\!\!=}^{(2)} \hspace{0.2cm}
\Gyr[\vb,\ub]
\begin{pmatrix}
\frac{\gamma_{(\vb\op\ub)\op\wb}^{\phantom{1}}}{\gwb} t
\\
\frac{\gamma_{(\vb\op\ub)\op\wb}^{\phantom{1}}}{\gwb} \{(\vb\op\ub)\op\wb \}t
\end{pmatrix}
\\[8pt]&
\overbrace{=\!\!=\!\!=}^{(3)} \hspace{0.2cm}
\begin{pmatrix}
\frac{\gamma_{(\vb\op\ub)\op\wb}^{\phantom{1}}}{\gwb} t
\\
\frac{\gamma_{(\vb\op\ub)\op\wb}^{\phantom{1}}}{\gwb} \gyrvub\{(\vb\op\ub)\op\wb \}t
\end{pmatrix}
\\[8pt]&
\overbrace{=\!\!=\!\!=}^{(4)} \hspace{0.2cm}
\begin{pmatrix}
\frac{\gamma_{\gyruvb\{(\vb\op\ub)\op\wb\}}^{\phantom{1}}}
{\gamma_{\gyruvb\wb}^{\phantom{1}}}
t
\\
\frac{\gamma_{\gyruvb\{(\vb\op\ub)\op\wb\}}^{\phantom{1}}}
{\gamma_{\gyruvb\wb}^{\phantom{1}}}
\gyrvub \{(\vb\op\ub)\op\wb\} t
\end{pmatrix}
\\[8pt]&
\overbrace{=\!\!=\!\!=}^{(5)} \hspace{0.2cm}
\begin{pmatrix}
\frac{
\gamma_{(\ub\op\vb)\op\gyruvb\wb}^{\phantom{1}}
}{
\gamma_{\gyruvb\wb}^{\phantom{1}}
}
t
\\
\frac{
\gamma_{(\ub\op\vb)\op\gyruvb\wb}^{\phantom{1}}
}{
\gamma_{\gyruvb\wb}^{\phantom{1}}
}
\{(\ub\op\vb)\op\gyruvb\wb \} t
\end{pmatrix}
\\[8pt]&
\overbrace{=\!\!=\!\!=}^{(6)} \hspace{0.2cm}
B(\ub\op\vb) \begin{pmatrix} t \\ \gyruvb\wb t \end{pmatrix}
\\[8pt]&
\overbrace{=\!\!=\!\!=}^{(7)} \hspace{0.2cm}
B(\ub\op\vb) \begin{pmatrix} t \\ \gyruvb\xb \end{pmatrix}
\\[8pt]&
\overbrace{=\!\!=\!\!=}^{(8)} \hspace{0.2cm}
B(\ub\op\vb) \Gyr[\ub,\vb] \begin{pmatrix} t \\ \xb \end{pmatrix}
\end{split}
\end{equation}
The chain of equations \eqref{guliv07} is valid for all
spacetime events $(t,\xb)^t$, $t\in\Rb$, $\xb=\wb t$, $\wb\in\Rct$,
thus verifying \eqref{guliv03}.

Derivation of the numbered equalities in \eqref{guliv07} follows:
\begin{enumerate}
\item\label{kamotb1}
Follows by definition, $\xb=\wb t$.
\item\label{kamotb2}
Follows from \ref{kamotb1} by a boost application to spacetime coordinates
according to \eqref{kyhd03}.
\item\label{kamotb3}
Follows from \ref{kamotb2} by
the definition of the spacetime gyration $\Gyr[\vb,\ub]$ in terms
of the space gyration $\gyrvub$ in \eqref{guliv01}.
\item\label{kamotb4}
Follows from \ref{kamotb3} by the identity
$\gwb=\gamma_{\gyruvb\wb}^{\phantom{1}}$, $\ub,\vb,\wb\in\Rct$,
that, in turn, follows from
the definition of gamma factors in \eqref{v72gs}
along with the invariance \eqref{eq005a}
of relativistically admissible velocities under gyrations.
\item\label{kamotb5}
Follows from \ref{kamotb4} by the linearity of gyrations along with
the gyrocommutative law of Einstein addition.
\item\label{kamotb6}
Follows from \ref{kamotb5} by a boost application to spacetime coordinates
according to \eqref{kyhd03}.
\item\label{kamotb7}
Follows from \ref{kamotb6} by definition, $\xb=\wb t$.
\item\label{kamotb8}
Follows from \ref{kamotb7} by
the definition of the spacetime gyration $\Gyr[\vb,\ub]$ in terms
of the space gyration $\gyrvub$ in \eqref{guliv01}.
\end{enumerate}
\end{proof}

The Boost Composition Theorem \ref{mainthm1} and its proof establish
the following two results:
\begin{enumerate}
\item\label{kamotd1}
The composite velocity of frame $\Sigma^{\prime\prime}$ relative to
frame $\Sigma$ in  Fig.~\ref{fig115} may, paradoxically, be
both $\ub\op\vb$ and $\vb\op\ub$.
Indeed, it is $\ub\op\vb$ in the sense that $\Sigma^{\prime\prime}$ is obtained from
$\Sigma$ by a boost of velocity $\ub\op\vb$ {\it preceded} by
the gyration $\Gyr[\ub,\vb]$ or, equivalently, it is obtained from
$\Sigma$ by a boost of velocity $\vb\op\ub$ {\it followed} by
the gyration $\Gyr[\vb,\ub]$.
\item\label{kamotd2}
The relationships \eqref{guliv02}\,--\,\eqref{guliv03} between boosts
and Thomas precession are equivalent to the
gyroassociative law and the gyrocommutative law
of Einstein velocity addition as we see from the proof of Theorem \ref{mainthm1}.
\end{enumerate}

In view of these two results of the Boost Composition Theorem,
the validity of the Thomas precession frequency, as shown graphically
in Fig.~\ref{fig115}, and the relationship between the Thomas precession angle $\ep$
and its generating angle $\theta$ stem from the
gyroassociative law of Einstein velocity addition.
Hence, in particular, the result that $\ep$ and $\theta$ have opposite signs
is embedded in the gyroassociative law of Einstein addition.
In the next section we will present a convincing numerical demonstration that
interested readers may perform to determine that, indeed,
$\ep$ and $\theta$ in Fig.~\ref{fig115} have opposite signs.

\section{Thomas Precession Angle and Generating Angle have Opposite Signs}

As in Fig.~\ref{fig115}, let $\ep$ and $\theta$ be
the Thomas Precession Angle and its generating angle, respectively.
As verified analytically, and as shown graphically in Fig.~\ref{fig115},
the angles $\ep$ and $\theta$ are related by \eqref{eq1d76s} and, hence,
they have opposite signs.

Without loss of generality, as in Fig.~\ref{fig115}, we limit our considerations
to two space dimensions.
Let $\ub,\vb\in\Rctwo$ be two nonzero relativistically admissible velocities
with angle $\theta$ between their directions, as shown in Fig.~\ref{fig115}.
Then, they are related by the equation
\begin{equation} \label{guliv08}
\frac{\vb}{\|\vb\|} =
\begin{pmatrix} \cos\theta&-\sin\theta\\\sin\theta&\phantom{-}\cos\theta \end{pmatrix}
\frac{\ub}{\|\ub\|}
\end{equation}

Let $\wb\in\Rctwo$ be the velocity of an object relative to frame
$\Sigma^{\prime\prime}$ in Fig.~\ref{fig115}.
Then, the velocity of the object relative to frame $\Sigma$ in Fig.~\ref{fig115}
is
\begin{equation} \label{guliv09}
\ub\op(\vb\op\wb) = (\ub\op\vb)\op\gyruvb\wb
\end{equation}
so that the velocity $\wb$ of the object is rotated relative to $\Sigma$
by the Thomas precession $\gyruvb$, which corresponds to the rotation angle $\ep$
given by \eqref{eq1d76s}. Hence,
\begin{equation} \label{guliv10}
\gyruvb\wb =
\begin{pmatrix} \cos\ep&-\sin\ep\\\sin\ep&\phantom{-}\cos\ep\end{pmatrix} \wb
\end{equation}
where $\ep$ is given by \eqref{eq1d76s}.

Substituting $\vb$ from \eqref{guliv08} into \eqref{guliv10}, we obtain the equation
\begin{equation} \label{guliv11}
\gyr[\ub, \frac{\|\vb\|}{\|\ub\|}
\begin{pmatrix} \cos\theta&-\sin\theta\\\sin\theta&\phantom{-}\cos\theta \end{pmatrix}
\ub]\wb
=
\begin{pmatrix} \cos\ep&-\sin\ep\\\sin\ep&\phantom{-}\cos\ep\end{pmatrix} \wb
\end{equation}

In \eqref{guliv11}, $\theta$ is the angle shown in the left part of Fig.~\ref{fig115},
which generates the Thomas precession angle $\ep$ shown in the
right part of Fig.~\ref{fig115},
where $\ep$ is determined by $\theta$ according to \eqref{eq1d76s} and, hence,
where $\theta$ and $\ep$ have opposite signs.
The validity of \eqref{guliv11} can readily be corroborated numerically.
The numerical corroboration of the validity of \eqref{guliv11}, in turn,
provides a simple way to convincingly confirm our claim that indeed
$\theta$ and $\ep$ have opposite signs.



\end{document}